\documentclass{sig-alternate-05-2015}
\pdfoutput=1
\usepackage{balance}
\usepackage[linesnumbered]{algorithm2e}
\usepackage[dvipsnames,table,xcdraw]{xcolor}
\usepackage{subscript}

\newcommand\ignore[1]{}
\newcommand{\schema}[1]{{\textbf{#1}}}

\usepackage[all]{nowidow}
\usepackage[colorinlistoftodos,prependcaption,textsize=tiny]{todonotes}
\usepackage{amsmath, amssymb, mathabx, colonequals, tikz, subfigure,
verbatim, float, xspace,url}
\def\qed {{\parfillskip=0pt \widowpenalty=10000 \displaywidowpenalty=10000
\finalhyphendemerits=0 \leavevmode \unskip \nobreak \hfil\penalty50
\hskip.2em \null \hfill $\square$ \par}} 
\usetikzlibrary{arrows,automata} 
\usepackage[noconfig]{ntheorem}

\usepackage[utf8]{inputenc}
\usepackage[T1]{fontenc}
\usepackage{graphicx}
\usepackage{tabularx}
\usepackage{listings}

\let\oldnl\nl%
\newcommand{\nonl}{\renewcommand{\nl}{\let\nl\oldnl}}
\newcommand{\midtilde}{\raisebox{-0.25\baselineskip}{\textasciitilde}}
\theoremstyle{definition}
\newtheorem{definition}{Definition}[section]

\theoremstyle{plain}
\newtheorem{theorem}[definition]{Theorem}
\newtheorem{lemma}[definition]{Lemma}
\newtheorem{proposition}[definition]{Proposition}

\newtheorem{example}[definition]{Example}

\input{macros.sty}
\usepackage[pdflatex]{gastex} 
\gasset{frame=false} 
\usepackage{tikz}

\begin{document}
%

\title{Functional Dependencies Unleashed\\ for Scalable Data Exchange}
%
%
%
%
%

\numberofauthors{3} 
%
\author{
%
%
\alignauthor
Angela Bonifati\\
     \affaddr{\fontsize{10}{12}\selectfont University of Lyon 1}\\
     \email{\fontsize{10}{12}\selectfont angela.bonifati@univ-lyon1.fr}
\alignauthor
Ioana Ileana\\
       \affaddr{\fontsize{10}{12}\selectfont  Paris Descartes University}\\
       \email{\fontsize{10}{12}\selectfont ioana.ileana@parisdescartes.fr}
\alignauthor 
Michele Linardi\\
       \affaddr{\fontsize{10}{12}\selectfont Paris Descartes University}\\
       \email{\fontsize{10}{12}\selectfont michele.linardi@parisdescartes.fr}
}

\maketitle
\begin{abstract}
We address the problem of efficiently evaluating target functional
dependencies (fds) in the Data Exchange (DE) process. Target fds naturally
occur in many DE scenarios, including the ones in Life Sciences in which multiple source relations need to be structured under a constrained target schema. However, despite their wide use, target fds' evaluation is still a bottleneck in the state-of-the-art DE engines. 
Systems relying on an all-SQL approach typically do not support
target fds unless additional information is provided. Alternatively, DE engines that do include these dependencies typically pay the
price of a significant drop in performance and scalability. In this paper, we present a novel chase-based algorithm that can efficiently handle arbitrary fds on the target. 
Our approach essentially
relies on exploiting the interactions between source-to-target (s-t) tuple-generating dependencies (tgds) and target fds. This allows us to tame the
size of the intermediate chase results, by playing on a careful ordering of
chase steps interleaving fds and (chosen) tgds. As a direct consequence, we
importantly diminish the fd application scope, often a central cause of the
dramatic overhead induced by target fds. Moreover, reasoning on dependency
interaction further leads us to interesting parallelization opportunities,
yielding additional scalability gains. We provide a proof-of-concept
implementation of our chase-based algorithm and an experimental study
aiming at gauging its scalability with respect to a number of parameters,
among which the size of source instances and the number of dependencies of
each tested scenario. Finally, we empirically compare with the latest DE engines, and show that our algorithm outperforms them.

\end{abstract}




\section{Introduction}

Data Exchange (DE) is the process of transferring data from a source database to a target database by leveraging the schemas of such databases, along with a set of high-level assertions involving schema elements, also called {\em dependencies}.
Dependencies in DE scenarios most often belong to two fundamental classes: (i) source-to-target (s-t) tuple-generating dependencies (tgds), which relate elements of the source schema with elements of the target schema; (ii) target dependencies, which involve solely elements of the target schema. 
Target dependencies can in turn be distinguished into target tgds, including target foreign keys, and target equality-generating dependencies (egds), including primary keys and, more broadly speaking, functional
dependencies (fds) on the target schema. 
DE scenarios pertain to a plethora of application domains. 
Due to the increasing size of available data, achieving performance and scalability on these scenarios is becoming overly stringent. 
Scientific applications in particular bring up a strong need for efficiency, due to both the abundance of data and the large number of dependencies they may yield. 
Among these dependencies, target fds are often pivotal in the tasks of integrating and exchanging data from such applications.

Research on DE has led to a very wide and rich range of results \cite{Clio, Fagin200589, FaginKP05}. 
Many existing approaches have aimed at {\em rewriting} the dependencies, so as to further leverage the power of an external, general evaluation engine, such as a traditional RDBMS \cite{Clio, MarnetteMPRS11} or a Datalog engine \cite{GreenKIT07,georgPods14}. Alternatively, custom chase engines have been used and optimized for computing DE solutions \cite{GeertsMPS14, PichlerS10}.

Despite the abundance of previous work in the area, the efficient support for target fds remains an open problem. 
Indeed, some of the fastest current DE engines such as Clio and ++Spicy \cite{Clio,MarnetteMPRS11} rely on rewriting approaches from DE dependencies into sets of SQL queries. 
While these engines show high efficiency when dealing with s-t tgds, their support for target fds is importantly limited. 
The reason for this stands in the fact that fds (and general egds) are difficult to enforce using a first-order (FO) query language such as SQL. 
Indeed, it was first conjectured in \cite{CateCKT09}, and later proved in \cite{MarnetteMP10}, that one cannot provide a complete FO rewriting for a DE scenario with target egds. 
To circumvent this limitation, previous work \cite{MarnetteMP10} has identified specific cases where rewriting with target fds may be achieved, typically by further exploiting a set of source dependencies.

However, what about DE scenarios including general target fds and no source dependencies? 
We argue that these scenarios are yet interesting in practice in many cases in which the presence of source dependencies cannot be guaranteed, e.g. when data sources are denormalized or lacking schema constraints.  
These may concern for instance Life Sciences applications, where massive, heterogeneous source data from scientific instruments, experiments or simply  observations of physical and biological events has to be integrated into a common unified schema with key constraints. 
While such scenarios are beyond the scope of all-SQL systems, they can be in principle handled by DE engines that depart from the all-SQL-induced limitations \cite{GeertsMPS14}. 
These engines can by design cover broader classes of constraints and have been shown to be quite efficient on a variety of scenarios. However, when it comes to evaluating target fds, they unfortunately exhibit significant drops in efficiency. 
This is strongly related to the fact that fd evaluation is an expensive task, and special methods and algorithms are needed to alleviate its negative impact on performance. 

In this paper, we address the problem of efficiently evaluating target fds in the DE process. 
To this end, we introduce a novel chase-based algorithm aimed precisely at achieving {\em efficiency and scalability in the presence of target fds}. 
We place ourselves in a DE setting where dependencies are either s-t tgds or target fds, and no source dependencies are available. 
To meet its performance goal, our algorithm will essentially leverage the dependencies' interactions. 
Guided by these interactions, our chase sequences will interleave fd and (chosen) tgd steps, with the aim of producing and handling piecemeal the so-called Saturation Sets. 

The key observation is that once a Saturation Set is constructed and processed, it can be safely added to the solution and removed from any further processing tasks. 
Structuring the chase by Saturation Sets allows us then to importantly lower the size of the intermediate chase results. 
This in turn importantly reduces the fd application scope, whose extent is one of the main reasons behind the major overhead induced by target fds. 
Moreover, reasoning in terms of dependency interaction and instance partitioning calls for very interesting parallelization opportunities, and thus scalability gains. 
The practical efficiency of our algorithm coupled with the range of scenarios it is able to cover let it stand an advantageous comparison with the latest, state-of-the-art DE engines \cite{GeertsMPS14, MarnetteMPRS11}, as shown in our experimental assessment. \\

\noindent{\bf Contributions.} The main contributions of our paper are the following:

\begin{itemize}
\item We propose a novel chase-based algorithm for efficient and scalable Data Exchange in the presence of arbitrary target functional dependencies. 
Our algorithm essentially exploits source-to-target and target dependencies interactions. 
Based on these, we infer a chase step ordering that allows to significantly speed up the evaluation of target fds. 
Our dependency interaction analysis also yields interesting parallelization opportunities, further boosting efficiency. 
We prove that our algorithm is correct and produces universal solutions. 

\item We provide a proof-of-concept implementation of our algorithm and an extensive experimental study aimed at gauging its performance and scalability with respect to a number of parameters, among which the size of source instances and the number of dependencies. 
The results of these experiments and the comparison with the fastest existing DE engines allow us to confirm and validate our performance goal. Moreover, they contradict the common wisdom stating that custom DE engines are less efficient than all-SQL ones.  
\end{itemize}

The paper is organized as follows. Preliminary definitions are in Section
2. The main algorithms are in Section 3 and 4. Experimental results are in
Section 5. A discussion of related work is in Section 6. We conclude in Section 7.

\section{Preliminaries}

\begin{figure*}
		\begin{picture}(170,60)(0,-3)
		\scalebox{0.7}{	\node(A)(20,74){
				\begin{tabular}{|c|c|c|}
				\multicolumn{3}{l}{\textbf{(i) Source Instance I }}\\
				\multicolumn{3}{l}{}\\
				\multicolumn{3}{l}{\textbf{Active\_Researcher}}\\
				\hline
				{\it name}  & {\it surname}  & {\it age} \\ \hline
				\rowcolor[HTML]{C0C0C0}
				Ronald & Red & 60 \\ \hline
				John & Gray & 33 \\ \hline
				\end{tabular}
			}}
			\scalebox{0.7}{\node(B)(29,13){
					\begin{tabular}{|c|c|c|c|}
					\multicolumn{4}{l}{\textbf{Awarded\_Researcher}}\\
					\hline
					{\it name}  & {\it surname}  & {\it awardName} & {\it year}  \\ \hline
					\rowcolor[HTML]{C0C0C0}
				         John & Gray & \prize &	2014 \\ \hline
					Wallace & Blue & \prize &	1932 \\ \hline
					\rowcolor[HTML]{C0C0C0}
					Fredric & Brown & \prize &	1932 \\ \hline
					Marlon & Bold & \prize &	1954 \\ \hline
					\rowcolor[HTML]{C0C0C0}
					Marlon & Bold  & \prize &	1972 \\ \hline
					\end{tabular}
				}}
				\scalebox{0.7}{\node(C)(29,42){
						\begin{tabular}{|c|c|c|c|}
						\multicolumn{4}{l}{\textbf{Researcher\_Collaboration}}\\
						\hline
						{\it $name_1$}  & {\it $surname_1$}  & {\it $name_2$} & {\it $surname_2$}  \\ \hline
						\rowcolor[HTML]{C0C0C0}
						Ronald & Red & Matthew & Orange \\ \hline
						Fredric & Brown & Miriam & White \\ \hline
						\end{tabular}
					}}
					
					\scalebox{0.7}{\node(D)(115,56){
							\begin{tabular}{|c|c|c|c|}
							\multicolumn{4}{l}{\textbf{(ii) Pre-solution J' (does not enforce egds)}}\\
							\multicolumn{4}{l}{}\\
							\multicolumn{4}{l}{\textbf{Researcher}}\\
							\hline
							{\it name}  & {\it surname}  & {\it idRewarding}  & {\it idClub}\\ \hline
							\rowcolor[HTML]{C0C0C0}
							Ronald & Red  & N\textsubscript{1} & N\textsubscript{2}\\ \hline
							\rowcolor[HTML]{d7e6ed}
							John & Gray  & N\textsubscript{3} & N\textsubscript{4}\\ \hline
							\rowcolor[HTML]{C0C0C0}
							John & Gray  & N\textsubscript{5} & N\textsubscript{6}\\ \hline
							\rowcolor[HTML]{d7e6ed}
							Wallace & Blue & N\textsubscript{7} & N\textsubscript{8}\\ \hline
							\rowcolor[HTML]{C0C0C0}
							Fredric & Brown & N\textsubscript{9} & N\textsubscript{10}\\ \hline
							\rowcolor[HTML]{d7e6ed}
							Marlon & Bold & N\textsubscript{11} & N\textsubscript{12}\\ \hline
							\rowcolor[HTML]{C0C0C0}
							Marlon & Bold & N\textsubscript{13} & N\textsubscript{14}\\ \hline
							\rowcolor[HTML]{d7e6ed}
							Ronald & Red & N\textsubscript{15} & N\textsubscript{16}\\ \hline	
							\rowcolor[HTML]{d7e6ed}
							Matthew & Orange & N\textsubscript{17} & N\textsubscript{16}\\ \hline
							\rowcolor[HTML]{C0C0C0}
							Fredric & Brown & N\textsubscript{18} & N\textsubscript{19}\\ \hline
							\rowcolor[HTML]{C0C0C0}
							Miriam & White & N\textsubscript{20} & N\textsubscript{19}\\ \hline
							\end{tabular}
						}}
						
						\scalebox{0.7}{\node(E)(115,12){
								\begin{tabular}{|c|c|c|}
								\multicolumn{3}{l}{\textbf{Research\_Prize}}\\
								\hline
								{\it awardName}  & {\it year}  & {\it idResearcher} \\ \hline
								\prize & 2014 & N\textsubscript{5} \\ \hline
								\prize  & 1932 & N\textsubscript{7} \\ \hline
								\rowcolor[HTML]{C0C0C0}
								\prize  & 1932 & N\textsubscript{9} \\ \hline
								\rowcolor[HTML]{d7e6ed}
								\prize  & 1954 & N\textsubscript{11} \\ \hline
								\rowcolor[HTML]{d7e6ed}
								\prize  & 1972 & N\textsubscript{13} \\ \hline
								\end{tabular}
							}}

							\scalebox{0.7}{\node(F)(198,64){
									\begin{tabular}{|c|c|c|c|c|}
									\multicolumn{4}{l}{\textbf{(iii) Solution J}}\\
									\multicolumn{4}{l}{}\\
									\multicolumn{4}{l}{\textbf{Researcher}}\\
									\hline
								{\it name}  & {\it surname}  & {\it idRewarding}  & {\it idClub}\\
									\rowcolor[HTML]{C0C0C0}	
									John & Gray  & N\textsubscript{5} & N\textsubscript{6} \\ \hline
									Wallace & Blue & N\textsubscript{7} & N\textsubscript{8} \\ \hline
									\rowcolor[HTML]{C0C0C0}	
									Marlon & Bold & N\textsubscript{13} & N\textsubscript{14}  \\ \hline
									Ronald & Red & N\textsubscript{15} & N\textsubscript{16} \\ \hline
									Matthew & Orange & N\textsubscript{17} & N\textsubscript{16}\\ \hline
									\rowcolor[HTML]{C0C0C0}	
									Fredric & Brown& N\textsubscript{7} & N\textsubscript{19}  \\ \hline
									Miriam & White& N\textsubscript{20} & N\textsubscript{19} \\ \hline
									\end{tabular}
								}}
								
								\scalebox{0.7}{\node(G)(196,22){
										\begin{tabular}{|c|c|c|}
										\multicolumn{3}{l}{\textbf{Research\_Prize}}\\
										\hline
											{\it awardName}  & {\it year}  & {\it idResearcher} \\ \hline
										\rowcolor[HTML]{C0C0C0}
										\prize & 2014 & N\textsubscript{5} \\ \hline
										\prize & 1932 & N\textsubscript{7} \\ \hline
											\rowcolor[HTML]{C0C0C0}
										\prize & 1954 & N\textsubscript{13} \\ \hline
										\prize & 1972 & N\textsubscript{13} \\ \hline
										\end{tabular}
									}}

									\end{picture}
									
\vspace{-0.3cm}
\caption{\label{label:fig1}Data
Exchange scenario involving researchers, their prizes and their collaborations.}
\end{figure*}

We assume three disjoint countably infinite sets, $\Deltac$ a set of
constants, $\Deltan$ a set of labelled nulls, and $\mathit{Var}$ a set of variables. 
A {\em term} is either a constant, or a labelled null, or a variable.
A \emph{database schema} $\schema S$ is a finite set of relation names each with fixed arity.
An \emph{atom} is a tuple of the form $R(t_1, \ldots, t_n)$ where every
$t_i$ is a term (also called {\em attribute}), $R$ is a relation name,
and $n$ is the arity of $R$.  
A \emph{fact} is an atom where all $t_i$ are constants or nulls. 
A \emph{database instance} of $\schema S$ is a set of facts
using relation names from $\schema S$. 
A homomorphism  between two database instances $K_1$ and $K_2$ \cite{Fagin200589} is a mapping $h$ from $\Delta_c \cup \Delta_n$ to $\Delta_c \cup \Delta_n$ such that: (i) $h(c) = c$, for every $c \in \Delta_c$; (ii) for every fact $R_i(\bar{t})$ of $K_1$, $R_i(h(\bar{t}))$ is a fact of $K_2$, where $h((a_1,..., a_n))$ = $(h(a_1), . . . , h(a_ n))$.
A homomorphism from $K_1$ to $K_2$ is said to be full if $h(K_1)=K_2$. 
An isomorphism between $K_1$ and $K_2$ is a full injective homomorphism $h$ from $K_1$ to $K_2$ such that $h^{-1}$ is a full injective homomorphism from $K_2$ to $K_1$. 
A \emph{functional dependency} (fd) is a constraint of the form $R.A \rightarrow R.B$ where $R$ is a relation of arity $n$ and $A,B$ are subsets of $\{1,\ldots,n\}$ indicating positions of attributes of $R$.\\

\noindent{\bf Data Exchange setting.} 
Given disjoint schemas $\schema S$ ({\em source schema}) and $\schema T$ ({\em target schema}), in this paper we focus on DE scenarios characterized by a quadruple (mapping) $M = (\schema S, \schema T,\Sigmast$,$\Sigmat)$, where  $\Sigmast$ is a set of s-t tgds and $\Sigmat$ is a set of fds on the target schema. 
An s-t tgd is a first-order (FO) logic formula of the form $\forall
\bar{x}\bar{y}. \phi(\bar{x},\bar{y}) \rightarrow \exists\bar{z} \psi(\bar{y},\bar{z})$, where $\phi(\bar{x},\bar{y})$ is a conjunction of atomic formulas over $\schema S$, and $\psi(\bar{y},\bar{z})$ is a conjunction of atomic formulas over $\schema T$. 
An fd $R.A \rightarrow R.B$ in $\Sigma_t$ is expressed as an egd by the formulas $\forall \bar{x} R(x_1, \ldots, x_n) \wedge
R(x'_1,\ldots,x'_n) \to \bigwedge_{i \in B} x_i = x'_i$, where in the left-hand side, $x_j$ and $x'_j$ are the same variable for all $j$ in $A$.  
We refer to the left-hand side of a dependency as its \textit{body}, while the right-hand side is the dependency's \textit{head}. 
Given $M = (\schema S, \schema T,\Sigmast$,$\Sigmat)$ and a database instance $I$ on $\schema S$ (source instance), we say that an instance $J$ on $\schema T$ is a \emph{solution for $I$ and $M$} if the instance $I \cup J$ satisfies all the logical constraints from $\Sigmast$ and $\Sigmat$, with the standard
meaning of satisfiability for first-order logic.\\

\noindent {\bf Assignments}. 
A \emph{variable assignment} (or simply an \emph{assignment}) over a domain $\mathit{V} \subseteq \mathit{Var}$ is a mapping $a: \mathit{V} \to \Deltac \cup \Deltan$, from variables in $\mathit{V}$ to constants or labelled nulls.
A \emph{body assignment} for a dependency is an assignment whose domain consists in the set of all body variables of the dependency. 
A \emph{full assignment} for a dependency is an assignment defined on the set of both body and head variables. 
We artificially extend assignments to constants by taking $a(c_i)=c_i$ for all assignments $a$ and all constants $c_i$. 
An instance $I$ is said satisfying a body assignment \textit{a} for a tgd $\forall \bar{x},\bar{y}$ $\phi(\bar{x},\bar{y}) \rightarrow
\exists\bar{z} \psi(\bar{y},\bar{z})$ if $I \models \phi(a(\bar{x}), a(\bar{y}))$. 
We equivalently define satisfiability for egd assignments (note that for these body assignments are also full). 
The {\em head materialization} of a full tgd assignment $a$ is the set of facts $\psi(a(\bar{y}),a(\bar{z}))$. \\

\noindent{\bf Data Exchange solutions and the chase}. 
DE solutions are usually obtained by employing the chase procedure. 
Several chase variants and flavors have been explored in the literature, in both general and DE specific settings. 
Chase procedures operate as sequences of chase steps, comprising {\em conditions} and {\em results} of application. 
A {\em terminating} chase sequence on an instance $I$ consists in iteratively applying chase steps until reaching a point where no chase step further applies or the sequence has failed. 
To obtain DE solutions via the chase procedure, one typically has to (i) issue a successful terminating chase sequence starting from the source instance, then (ii) retain the part of the result corresponding to the target schema (hence, the target instance). 
Depending on the constraints, the source instance and the specific chase flavor, issues regarding termination may arise. 
However, if obtained, solutions generated by chasing enjoy an intrinsic property called the \emph{universality} \cite{Fagin200589}. 
We say that a solution $J$ for a mapping $M$ and a source instance $I$ is
universal ($J \in USol(M,I)$) iff for every solution $K$ there is a
homomorphism from $J$ to $K$. \\

\noindent{\bf The Oblivious Chase}. 
In this paper, we rely on one of the simplest chase variants, called the Oblivious (or Naive) Chase \cite{GrahneO14}.
We briefly describe below Oblivious Chase steps:
\begin{itemize} 
\item \textbf{Tgds.} 
An Oblivious Chase step with a tgd $t$ of the form $\forall \bar{x},\bar{y}$ $\phi(\bar{x},\bar{y}) \rightarrow
\exists\bar{z} \psi(\bar{y},\bar{z})$ applies on an instance $I$ if there exists a body assignment $a$ of $t$ such that (i) $I$ satisfies $a$ and (ii) $a$ has not already been used for a previous chase step with $t$ in the chase sequence. 
The result of applying the chase step is obtained by (i) extending $a$ to a full assignment by injectively assigning to each $z_i \in \bar z$ a fresh null in $\Delta_n$ and then (ii) adding the facts in $\psi(a(\bar{y}),a(\bar{z}))$ to $I$.  
\item \textbf{Egds}. 
An Oblivious Chase step with an egd $e$ of the form $\forall \bar{x}$ $\phi(\bar{x}) \rightarrow (x_{i}=x_{j})$ applies on an instance $I$ if there exists an assignment $a$ of $e$ such that (i) $I$ satisfies $a$ and (ii) $a(x_{i}) \neq a(x_{j})$. To define the application result of the chase step we distinguish three cases:
(i) if both $a(x_{i})$ and $a(x_{j})$ are constants, the chase fails; (ii) if exactly one of $a(x_{i})$, $a(x_{j})$ is a null, $I$ is modified by replacing all the occurrences of the null with the constant; (iii) if both $a(x_{i})$ and $a(x_{j})$ are nulls, $I$ is modified as in (ii) by picking one null at random \footnote{Or, one can assume a total ordering on labelled nulls and always choose the smallest one, as done in an alternative definition of the Oblivious Chase in \cite{CaliGK13}, that could also be seamlessly adopted here.}.
\end{itemize}

For solutions obtained by employing the Oblivious Chase in our DE setting comprising s-t tgds and target fds, no termination issues arise and the following holds \cite{GrahneO14}:

\begin{proposition}\label{prop:iso}
If one Oblivious Chase sequence fails then all sequences fail and there are no solutions for the DE scenario.

Else, a solution produced by the Oblivious Chase is universal and, moreover, unique up to isomorphism.
\end{proposition}

\noindent \textbf{Pre-solutions and the Classical DE Chase}. A customary procedure for obtaining solutions in our DE setting consists in retaining only a subset of the possible solution-generating Oblivious Chase sequences. 
Namely, those that first chase with the s-t tgds, then with the target egds. 
This allows restricting the result to the target schema earlier in the process (after concluding the chase with the tgds), thus obtaining a {\em pre-solution}. 
We will hereafter refer to this chase variant as the {\em Classical Data Exchange Chase}.

\begin{example}
  \label{ex:running-example}
  Consider a DE scenario where $\cal M$ = ($\schema S$, $\schema T$,
  $\Sigmast$,$\Sigmat)$ with $\schema S$ and $\schema T$ illustrated in Figure~\ref{label:fig1}, and $\Sigmast$ and $\Sigmat$ as below. We omit the quantifiers, and use uppercase letters for
  existentially quantified variables and lowercase ones for universally quantified variables.
  \vspace{0.2cm}\\
  $\begin{array}{l}
    \mathbf{m_1}: Active\_Researcher(n,s,a)  \rightarrow Researcher(n,s,Y_1,Y_2) \\
\end{array}$\\
 $\begin{array}{l}
\mathbf{m_2}: Awarded\_Researcher(n',s',p',w')  \rightarrow 
\end{array}$\\
 $\begin{array}{l}
  \ \ Researcher(n',s',T,T_1) \wedge Research\_Prize(p',w',T) 
 \end{array}$\\
$\begin{array}{l}
\mathbf{m_3}: Researcher\_Collaboration(n'',s'',n''',s''')  \rightarrow 
\end{array}$\\
$\begin{array}{l}
 \ \ Researcher(n'',s'',E_1,E_2) \wedge Researcher(n''',s''',E_3,E_2)
\end{array}$\\
 $\begin{array}{l}
 \mathbf{e_1}: Researcher(n,s,p,w) \wedge Researcher(n,s,p',w') \rightarrow 
 \end{array}$\\
  $\begin{array}{l}
  \ \ \ (p = p') \wedge (w = w')
  \end{array}$\\
 $\begin{array}{l}
 \mathbf{e_2}: Research\_Prize(p,w,z) \wedge Research\_Prize(p,w,z')\rightarrow
 \end{array}$\\
 $\begin{array}{l}
   \ \ \ (z = z')
 \end{array}$\\

Figure ~\ref{label:fig1}(i) depicts the source instance $I$. By first chasing with the s-t tgds $m_1$, $m_2$, $m_3$ we obtain the {\em pre-solution} $J'$ in Figure \ref{label:fig1}(ii), where $N_1$-$N_{20}$ are fresh labelled nulls over the relations $Researcher$ and $Research\_Prize$. Next, by applying the egds (primary keys) $e_1$ and
$e_2$, we obtain the solution $J$ in Figure \ref{label:fig1}(iii).

\end{example}

\section{Optimizing the chase}

\begin{figure*}[t]
\begin{small}
\begin{center}
\begin{tabular}{|l|l|}\hline
{\em s-t tgd}
& {\em full tgd assignments} \\\hline\hline
$m_1$ & $a_1m_1$ = \{$n:Ronald$, $s:Red$, $a:60$, $Y_1:N_1$, $Y_2:N_2$\} \\
 & $a_2m_1$ = \{$n:John$, $s:Gray$, $a:33$, $Y_1:N_3$, $Y_2:N_4$\} \\\hline
$m_2$ & $a_1m_2$=\{$n':John$, $s:Gray$, $p':\prize$, $w':2014$, $T:N_5$, $T_1:N_6$\}\\
 & $a_2m_2$=\{$n':Wallace$, $s:Blue$, $p':\prize$, $w':1932$, $T:N_7$, $T_1:N_8$\} \\
 & $a_3m_2$=\{$n':Fredric$, $s:Brown$, $p':\prize$, $w':1932$, $T:N_9$, $T_1:N_{10}$\} \\
 & $a_4m_2$=\{$n':Marlon$, $s:Bold$, $p':\prize$, $w':1954$, $T:N_{11}$, $T_1:N_{12}$\} \\
 & $a_5m_2$=\{$n':Marlon$, $s:Bold$, $p':\prize$, $w':1972$, $T:N_{13}$, $T_1:N_{14}$\} \\\hline
$m_3$ & $a_1m_3$ = \{$n'':Ronald$, $s'':Red$, $n''':Matthew$, $s''':Orange$, $E_1:N_{15}$, $E_2:N_{16}$, $E_3:N_{17}$\} \\
 & $a_2m_3$ = \{$n'':Fredric$, $s'':Brown$, $n''':Miriam$, $s''':White$, $E_1:N_{18}$, $E_2:N_{19}$, $E_3:N_{20}$\} \\\hline
\end{tabular}
\end{center}
\end{small}
\vspace{-0.3cm}
\caption{Initial set of full assignments for the s-t tgds of Example \ref{ex:running-example}.}\label{fig-assign}
\vspace{-0.2cm}
\end{figure*}

Not surprisingly, a direct implementation of the Classical Data Exchange Chase procedure described above, building the pre-solution before proceeding to egd application, will in general fail to achieve practical performance. The main reason for this is that egd application is a time-intensive, resource-consuming task; the larger the pre-solution, the more expensive the final egd application will get. On the other hand, the Oblivious Chase allows for much more flexibility in the choice of the order of chase steps. However, it is not obvious which of the Oblivious Chase sequences may end up performing better than others. 

In the following, we will show that some Oblivious Chase sequences are particularly interesting, as they can lead to increased performance by providing means for targeted optimizations. We will introduce a new chase-based algorithm, called the \ouralgo, that produces and leverages such sequences. While the \ouralgo\ relies on Oblivious Chase steps, it will aim at improving performance by {\em exploiting the dependencies' interaction} to infer chase orders that allow systematically {\em reducing the egd application scope}. 

In this section, we will gradually expose key ideas and intuitions as well as define the main concepts and tools used by our algorithm, namely assignment lifecycle, interactions and collisions, Saturation Sets and overlapping. We will then proceed to the detailed description of the \ouralgo\ in Section 4.

\subsection{Chase and assignment lifecycle}

In order to illustrate the ideas behind the \ouralgo, we first invite the reader to adopt a slightly different point of view on the Oblivious Chase in our Data Exchange setting, namely by considering chase steps as operations that consist in {\em picking and modifying full tgd assignments}.
Indeed, in our setting, the set of body assignments 
for chasing with tgds only depends on the s-t tgds and the source instance and is thus fixed and known in advance. This also holds for full tgd assignments, which are basically extensions of body s-t tgd assignments with fresh labeled nulls.  


The chase will thus bootstrap on an {\em initial set} of full s-t tgd assignments containing all full tgd assignments for our scenario. For convenience, we add a unique identifier to each of these assignments. Each s-t tgd step will then {\em pick} one of these initial full assignments, that is, remove the assignment from the initial set and add it to a {\em target set} of assignments. 

What about egd steps? An egd step will simply consist in {\em modifying} full tgd assignments that have previously been added to the target assignment set, replacing values in the image of these assignments. 

Upon chase termination, the initial set of assignments will be empty whereas the target assignment set will comprise all the full tgd assignments originally present in the initial set (by their identifier), potentially modified because of egd application. Intuitively, a chase sequence can be thus seen as developing, for each element in a fixed known initial set of full tgd assignments, a {\em lifecycle} which starts with the assignment being picked (removed from the initial set and added to the target set) and continues with the assignment evolving in the target set, due to egd application, up to a point where it becomes immutable.
Throughout the lifecycles of assignments, and, precisely, at each chase step, the intermediate target instance can be obtained by performing the head materialization of the assignments in the the target assignment set. The solution will be given by the head materialization of the immutable forms of all the assignments, that is, after the termination of a chase sequence.

\begin{example}
Consider again the scenario in Example \ref{ex:running-example}. The initial set of full s-t tgd assignments is illustrated in Figure \ref{fig-assign}.

A possible Oblivious Chase sequence can begin with the following steps:
\begin{itemize}
\item chase step with $a_2m_1$ ({pick assignment}), adding the assignment $a_2m_1$ to the target set and removing it from the initial set;
\item chase step with $a_1m_2$ ({pick assignment}), adding the assignment $a_1m_2$ to the target set and removing it from the initial set; 
\item chase step with $e_1$ ({modify assignment}), which will transform $a_2m_1$ into $a_2m_1$=\{$n:John$, $s:Gray$, $a:33$, $Y_1:N_5$, $Y_2:N6$\} in the target set. 
\end{itemize}

\end{example}\label{ex-assign}

In the remainder, unless otherwise specified, we use assignment as a short for full s-t tgd assignment. We will say that two assignments are distinct if their identifiers are distinct. 
We will further employ the terms {\em assignment set} to refer to all the assignments for a given scenario, regardless of their lifecycle status. For our example, the assignment set is thus \{$a_1m_1$, $a_2m_1$, $a_1m_2$, $a_2m_2$, $a_3m_2$, $a_4m_2$, $a_5m_2$, $a_1m_3$, $a_2m_3$\}. We will refer to an assignment set in its initial form (when bootstrapping the chase) as the initial assignment set and to intermediate sets produced when chasing as (intermediate) target assignment sets.

\subsection{Assignment interaction}

Thanks to our new view on the chase, we further proceed to characterizing the cases when two assignments "are involved together in some egd (fd) application", that is, when the materialization of these assignments creates the application conditions for an egd. We will characterize this situation via the notion of {\em assignment interaction} as follows:

\begin{definition}[Assignment interaction]
Two (non\--ne\-cessarily distinct) assignments $a_1$ and $a_2$, in some form induced by some chase sequence, are said to interact on a target fd $f$ of the form $R.A \rightarrow R.B$ and term pairs $<v_1$, $v'_1>,\dots,<v_n, v'_n>$ (where $n$ is the cardinality of $A$) iff, denoting by $m_1$ and $m_2$ the s-t tgds for these assignments:
\begin{itemize}
\item the terms $v_i$ and $v'_i$ occur in the $A$ positions of an $R$ atom in the heads of $m_1$ and $m_2$ respectively;
\item $a_1(v_i) = a_2(v'_i)$ for all $i$.
\end{itemize} 
\end{definition}

We move now to a key observation underlying our approach to performance improvement: if at some point during a chase sequence one can somehow guarantee that remaining assignments (those not yet picked) will never interact with those already picked, egds can be applied to termination and the currently maintained target set of assignments {\em saved as is to the solution\footnote{By taking their head materializations thereof.}}, and then {\em flushed}, discarded entirely. Then, the rest of the chase sequence will proceed starting with an empty target set.
The direct consequence of this is that of {\em diminishing the size of the intermediate result}, and thus {\em reducing the egd application scope}. 

\begin{example}\label{ex-saveflush}
Consider again the chase steps in Example \ref{ex-assign}. One can verify that after the first two chase steps with $a_2m_1$ and $a_1m_2$ respectively, there exists no possible chase sequence that would make any of the remaining assignments interact with $a_2m_1$ or $a_1m_2$.  

After the chase step with $e_1$, no other egd applies to our intermediate result, comprising $a_2m_1$=\{$n:John$, $s:Gray$, $a:33$, $Y_1:N_5$, $Y_2:N6$\} and $a_1m_2$=\{$n':John$, $s':Gray$, $p':\prize$, $w':2014$, $T:N_5$, $T_1:N_6$\}. These two assignments can be saved to the solution and discarded from our current target set.
\end{example}

We can repeat this "save and flush" procedure again and again until the end of our chase sequence, thus maintaining at each chase step only part of the 
intermediate result. To achieve this however we need to be able to detect the possible interactions between assignments well in advance, since interaction is chase sequence and lifecycle dependent.  
We will see in the next subsection that this is possible by defining a stronger notion, namely the {\em collision} between assignments.

\subsection{Collisions and Saturation Sets}

A very interesting property (consequence of the definition of the Oblivious Chase) is that determining whether two assignments will interact at some point during their lifecycle is {\em chase sequence independent}. One can indeed show that the following holds\footnote{By an isomorphism argument.}:

\begin{proposition}\label{prop:collision_fixed}
Let $C_1$ and $C_2$ be two arbitrary terminating chase sequences and $a_1$ and $a_2$ be two (non-necessarily distinct) assignments. Then $a_1$ and $a_2$ interact during $C_1$ iff they interact in their final, immutable form after $C_1$ has ended iff they interact in their final, immutable form after $C_2$ has ended. 

If any of the equivalent statements above holds, we say that $a_1$ and $a_2$ {\em collide}.
\end{proposition}

The result above provides us with the fixed, lifecycle and chase-sequence independent {\em collision} relation on assignments. Using the notion of collision we can then define {\em groups of assignments that do not collide with the rest of assignments}, that is, Saturation Sets:

\begin{definition}[Saturation Sets]
Given the assignment set $A$ for a Data Exchange scenario, a Saturation Set $S$ is a subset of $A$ such that for all assignments $a \in S$ and $a' \in A-S$, $a$ and $a'$ do not collide.
\end{definition} 

Notice that Saturation Sets can be carved out from the initial set of assignments and can evolve into (intermediate) target results later on during the lifecycle stages. Since collision is lifecycle and chase-sequence independent however, our definition holds without referring to a specific form / lifecycle stage of assignments. In our previous Example \ref{ex-saveflush}, it turns out that we were able to apply the "save and flush" strategy at a specific moment during the chase mainly because we have constructed, as an (intermediate) target result, a Saturation Set. Our goal will be thus intuitively that of structuring the chase sequence so as to group together assignments belonging to the same Saturation Set. This will allow us in turn to repeatedly apply the "save and flush" strategy and reduce the size of the intermediate results.\\ 

\noindent \textbf{Saturation Set size}. Note that the above definition of Saturation Sets is very permissive and does not tell us how to build such sets. Indeed, it is enough to note that if $S_1$ and $S_2$ are two Saturation Sets, then their union is also a valid Saturation Set by our definition; as it turns out, the whole assignment set is also a valid Saturation Set. As a consequence, finding Saturation Sets and structuring the chase accordingly by interleaving "save and flush" operations is not directly a guarantee of optimization. Importantly, what we aim for are {\em comparably smaller Saturation sets}. 

In principle, in order to find ideally small candidates for such Saturation Sets, one would need to compute the graph of collision, which would in turn imply running the chase sequence until completion. In the following, we present the alternative concept of {\em overlap}, which is coarser but easier-to-compute than the notion of collision. We will show how the overlap lets us explore the search space of possible Saturation Sets in the next subsection. 

\subsection{Overlaps and overlapping assignments}

To be able to define overlap between assignments, we start by introducing an essential notion: the mutable existential variables (or mutable existentials).

\begin{definition}[Mutable existential variable]
Given\\an s-t tgd $m$ we inductively define the mutable existential variables in the head of m as follows: 

For every fd $f$ of the form $R.A \rightarrow R.B$ and every atom $r$ of relation $R$ in the head of $m$, denoting by $(v_1, \dots, v_n)$ the terms appearing in the $A$ positions of $r$, if the following hold:
\begin{itemize}
\item all $v_i$ are either universal variables, constants or mutable existential variables
\item or there exists $r'\neq r$ an atom of relation $R$ in the head of $m$ s.t., denoting by $(v'_1, \dots, v'_n)$ the terms occurring in the $A$ positions in $r'$, for all $i$, either
\begin{itemize} 
\item $v_i$ and $v'_i$ are universals, constants, or mutable existentials
\item or $v_i=v'_i$
\end{itemize}

then any existential variable appearing in any of the $B$ positions in $r$ is mutable.
\end{itemize}
\end{definition}

Intuitively, the notion of mutable existentials allows for a rough account of possible value changes caused by egd application: a null introduced by a mutable existential can be seen as {\em prone to change} due to one or several egd steps. Note also that the definition above is well-formed:  the operational process it describes for iteratively identifying mutable existentials is indeed finite because of the finite number of existential variables in any tgd's head.

\begin{example}
Consider the tgd (existentials in capital)\\ $A(x,y)\rightarrow$ $R(x,C)$, $S(C,G)$, $T(y,D)$, $U(L,D)$, $W(V,M)$, $W(V,N)$ and further assume that all first attributes are keys for the relations $R$, $S$, $T$, $U$, $W$.
Then $C$ is mutable because of the presence of $x$ (universal) on the first position in $R$ and the key on $R$. The same holds for $D$ because of $y$ in the T atom. Also, $G$ is mutable  because of the presence of $C$ (mutable) on the first position in S. With the second part of the definition, $M$ and $N$ are mutable because of $V$ occurring on the first position in the two $W$ atoms. On the other hand, $L$ is not mutable since it does not fall under any of the definition's conditions.
\end{example}  

Using the above notion of mutable existentials, we can now define overlap and overlapping assignments as follows:

\begin{definition}[Overlapping assignments]
Two {\em distinct} assignments $a_1$ and $a_2$ are said overlapping on a target fd $R.A \rightarrow R.B$ and term pairs $<v_1$, $v'_1>,\dots,<v_n, v'_n>$ (where $n$ is the cardinality of $A$) iff, for all $i$, and denoting by $m_1$ and $m_2$ the s-t tgds for $a_1$ and $a_2$ respectively:
\begin{itemize}
\item $v_i$ and $v'_i$ occur in the $i$th $A$ positions of an $R$ atom in the heads of $m_1$ and $m_2$ respectively and
\item both $v_i$ and $v'_i$ are either universal variables, constants or mutable existentials and
\item if both $a_1(v_i)$ and $a_2(v'_i)$ are constants then $a_1(v_i)$=$a_2(v'_i)$
\end{itemize} 
\end{definition}

\begin{example}
It is easy to see that the assignments $a_2m_1$ and $a_1m_2$ of Example 3.3 are overlapping, on term pairs $<n,n'>$ and $<s,s'>$ and fd $e_1$.    
\end{example} 

As is the case for interaction, the overlap notion is lifecycle-dependent. As opposed to interaction however, lifecycle evolution intuitively refines overlapping {\em in the opposite direction}: that is, two assignments can overlap in their initial form but while evolving during the chase they {\em may end up non-overlapping}. Thus, it is interesting to delay overlap evaluation, or alternatively, to proceed to early egd application. We will provide more details on this in Section 4.  


As sketched above, our goal will be to employ the notion of overlap instead of the notion of collision. Indeed, our algorithm will exploit the notion of overlap in order to build Saturation Sets. The following result is crucial for our solution computation procedure by showing that we are indeed correct in doing so, because collision always (i.e. no matter the lifecycle stage) implies overlap:

\begin{proposition}\label{prop:overlap_collision}
Let $a_1$ and $a_2$ be two {\em distinct} assignments such that $a_1$ and $a_2$ collide.

Then at any point of their lifecycle $a_1$ and $a_2$ are overlapping.
\end{proposition}

\section{A chase-based algorithm}

In the following, we illustrate our chase-based algorithm, the \ouralgo. We first present the main algorithmic bricks, then focus on refinements and optimizations.

\subsection{Main algorithmic bricks}

Algorithm \ref{mainloopv0} sketches the global workflow of the \ouralgo. Our algorithm will start by computing the initial assignment set $A_m$ for each s-t tgd $m$ (recall that the initial assignment set for a DE scenario consists in the union of those). It will then repeatedly: (i) construct a Saturation Set and chase this set to termination and (ii) add the result (i.e. the obtained final Saturation Set) to the target solution
by applying head materialization to the assignments of the Saturation Set. 
This procedure (lines \ref{chooseseed} - \ref{addsat}) is repeated until the entire assignment set has been consumed, i.e. partitioned in Saturation Sets (line \ref{bigloop}).

\begin{algorithm}[h!]
\nonl \textbf{Algorithm} Interleaved Chase\\
\nonl {\em\ \ \ Can raise Exception:ChaseFail}\\
\KwIn{$M = (S,T,\Sigma_{st},\Sigma_{t})$,\textbf{SrcInstance} $I$} 
\KwOut{\textbf{TgtInstance} $J$: the target solution}

$J=\emptyset$\;
\ForEach {$m \in \Sigmast$}
{\label{asgbyrule}
	$A_m$ = InitialAssignmentSet$(m,I)$\;
}
\While{$\exists m \in \Sigmast$ s.t. $A_m \neq \emptyset$}
{\label{bigloop}
	$seed$ = SelectArbitraryAsg($A_m$)\;\label{chooseseed}
	$\Ssat$ = BuildAndChaseSaturationSet($seed$)\;	 \label{buildsat}
    $J = J \cup$ Materialize($\Ssat$)\;\label{addsat}
}\label{endsatdesc}       
\Return $J$\;
\caption{Main loop of the \ouralgo \label{mainloopv0}}
\end{algorithm}

Each new Saturation Set construction relies on first selecting an arbitrary assignment that has not yet been attributed to some previous Saturation Set (line \ref{chooseseed}). The bulk of our solution computation procedure then relies on the subroutine \-$BuildAndChaseSaturationSet$. Algorithm \ref{buildsatv0} shows an initial, unrefined form of this subroutine.\\

\begin{algorithm}[h!]
\nonl \textbf{Subroutine} BuildAndChaseSaturationSet\\
\nonl {\em\ \ \ Can raise Exception:ChaseFail}\\
\KwIn{Seed Assignment $seed$} 
\nonl \textbf{Globals: $A_m \forall m \in \Sigmast$}\\
\KwOut{Saturation Set $\Ssat$}

$NewAssignments$ = $\{seed\}$\;
$\Ssat = \{seed\}$\;
$A_{Tgd(seed)} -=\{seed\}$\;
ApplyEgdsToTermination($\Ssat, seed$)\;
\While{$NewAssignments \neq \emptyset$}
{ 
	$a$ = SelectArbitraryAsg($NewAssignments$)\;
    $NewAssignments - = \{a\}$\;
    \ForEach {$a'\in \bigcup_{m\in \Sigmast}^{\ }{A_m}$ s.t. $a$ and $a'$ overlap}
    {\label{expandsat}
        $NewAssignments \cup= \{a'\}$\; 
        $\Ssat \cup=\{a'\}$\;
        $A_{Tgd(a')}-=\{a'\}$\;
        ApplyEgdsToTermination($\Ssat, a'$)\;\label{applyegds}
       
    } 
}
\Return $Ssat$\;
\caption{BuildAndChaseSaturationSet\label{buildsatv0}, initial}
\end{algorithm}

\noindent \textbf{Saturation Set construction}. The $BuildAndChaseSat\-urationSet$ subroutine initializes the current Saturation Set with the provided seed assignment. It then iteratively expands this Saturation Set by embedding new assignments that have at least an overlap with previously added assignments. This is done until a fixpoint is reached, i.e. no assignment can be further added.\\

\noindent \textbf{Chase steps}. The $BuildAndChaseSaturationSet$ subroutine also fires the chase steps on the Saturation Set.
Tgd chase steps are in reality subsumed by Saturation Set construction. They indeed merely consist in picking assignments, removing them from the initial assignment set\footnote{We assume a zero-cost call Tgd($a$) that allows us to retrieve the s-t tgd corresponding to an assignment; this can be seen as accessing an attribute of the assignment.} and adding them to the current Saturation Set. After each such addition of an assignment to the current Saturation Set, a sequence of egd chase steps is further applied to termination. This is done via the $ApplyEgdsToTermination$ subroutine, namely by firing all the egds triggered (directly or indirectly) by the assignment last added to the current Saturation Set. Note that this method can raise a {\em ChaseFail} exception corresponding to cases when the chase fails upon egd application attempt. This exception will be propagated through our algorithm to induce a global {\em ChaseFail} exit. If no failure occurs during the egd sequences, the $BuildAndChase\-SaturationSet$ subroutine will return the final Saturation Set, chased to termination, and ready to be converted into part of the target solution.


Note that, in the above egd chase steps, we have opted for an {\em early egd application} strategy that leads to firing egds as soon as applicable. As briefly mentioned in Section 3, this allows us to reduce the overall size of the Saturation Set, by leveraging a refined overlap among assignments. We show hereafter an example of how such size reduction occurs:

\begin{example}
Consider a Data Exchange scenario comprising unary source relations $A$ and $B$ and binary target relations $R$ and $S$, along with a source instance consisting in the tuples $A(1)$, $B(2)$ and a set of s-t tgds as follows: $m_1: A(x) \rightarrow R(x,y),S(y,z)$, $m_2: A(x) \rightarrow R(x,x)$ and $m_3: B(x) \rightarrow S(x,z)$. Furthermore, two functional dependencies (keys on first attributes) $f_1 : R.\{1\}\rightarrow R.\{2\}$ and $f_2 : S.\{1\}\rightarrow S.\{2\}$ are defined on the target. 

The initial assignment set comprises: $a_1m_1$ = \{$x:1$, $y:N_1$, $z:N_2$\}, $a_1m_2$ = \{$x:1$\} and $a_1m_3$ = \{$x:2$, $z:N_3$\}.

Assume that we start constructing a Saturation Set with $a_1m_1$ as seed. We then retrieve $a_1m_2$ because it is overlapping with $a_1m_1$ on the pair of terms $<x,x>$ and $f_1$. We have now two options. The first is to not apply any egd, and continue by adding $a_1m_3$ to our current Saturation Set, since it overlaps with $a_1m_1$ on the pair of terms $<y,x>$ and $f_2$ (note that $y$ is a mutable existential). We end up with a Saturation Set containing 3 assignments.

The second option is to apply egds prior to evaluating overlapping with $a_1m_3$. This will in turn change the assignment $a_1m_1$ to \{$x:1$, $y:1$, $z:N_2$\}. Because of this change, $a_1m_1$ and $a_1m_3$ are no longer overlapping; indeed, the previous overlap on $<y,x>$ no longer holds. Due to early egd application, the Saturation Set ends up containing only two assignments, $a_1m_1$ and $a_1m_2$, instead of the previous 3. 
\end{example} 

Moreover, through consistent early application of the egds, we can seamlessly guarantee that only egd application conditions created by the assignment last introduced 
need to be checked, instead of re-examining our whole intermediate result.
We do not further detail the implementation of $ApplyEgdsTo\-Termination$, mainly because any implementation 
can be plugged in as long as it proceeds to the necessary changes on
assignments, and further raises a {\em ChaseFail} exception in case of
failed step attempts. 
Recall also that the main purpose of our algorithm is that of reducing the
intermediate result on which egd steps will be applied, and this is
orthogonal and complementary to egd-specific optimizations, that can be 
envisioned as future work.\\


\noindent \textbf{New assignment discovery}. A core task of the $BuildAnd\-ChaseSaturationSet$ subroutine is the iterative discovery of overlapping assignments, thus progressively expanding the current Saturation Set. A baseline take on assignment discovery would consists in repeatedly considering all assignments in the current Saturation Set and for each of those searching for overlapping assignments still available. 

We include in our initial algorithm presentation a straightforward optimization from this baseline. Namely, we only consider a {\em frontier} of our Saturation Set, consisting in the newly added assignments that have not yet lead to expansion via overlapping. This frontier is represented by the $NewAssignments$ variable, which holds the set of assignments "yet to be verified". We dedicate the following subsection to presenting more complex, additional designs that further improve the performance in assignment discovery.

\ignore{
  \SetKwFunction{mainloop}{}
  \SetKwProg{myalg}{InterleavedChase}{}{}
  \myalg{\mainloop}{

   xxx\;
   \SetKwFunction{proc}{BuildSaturationSet}
   \proc{}\;
   \KwRet\;}
  {}
   \caption{bla}
\end{algorithm}
}

\ignore{
\begin{algorithm}
	\KwIn{$M = (S,T,\Sigma_{st},\Sigma_{t})$,\textbf{SrcInstance} $I$} 
	\KwOut{\textbf{TgtInstance} $J$: the target solution}
    $J=\emptyset$\;
    $A= AssignmentSet(M, I)$\;
	\While{$A \neq \emptyset$}
	{
	      //\textit{build and process a new Saturation Set}\\
	      $Ssat = \emptyset$\;
          seed = PickArbitraryAsg($A$)\;
          $NewAssignments$ = $\{seed\}$\;
          DoChaseSteps($seed$,$\Ssat$)\;
          \While{$NewAssignments \neq \emptyset$}
          { $a$ = PickArbitraryAsg($NewAssignments$)\;
            $NewAssignments - = a$\;
            \ForEach {$a'\in A$ s.t. $a$ and $a'$ overlap}
            {
                $NewAssignments \cup= \{a'\}$\; 
          		DoChaseSteps($a'$,$\Ssat$)\;
          	} 
          }

    //\textit{add result to the target instance}\\		 
    $J = J \cup \Ssat$\;\
    }		          
    \Return $J$\;
    \caption{The DataExchange Algorithm \label{algv0}}
\end{algorithm}
}

\ignore
{
\noindent \textbf{Saturation Sets and materialized assignments}. Recall that our view on the chase is one playing on assignments: tgd steps will consume assignments whereas egd steps will make these assignments evolve, the overall process inducing what we call the assignments lifecycle (Section ...). Once the chase has terminated, a final, immutable form of the assignments can be materialized to obtain the target instance.

For practical purposes we choose to explicitly keep the materialized form of the assignments and thus to consider the saturation set as comprising both assignments and atoms, which we refer to as $\Ssat.assignments$ and $\Ssat.atoms$ respectively. We moreover consider that for any atom $a$ in $\Ssat.atoms$ a (zero-cost) call $assignment(a)$ provides the assignment that, by its materialization, yields $a$.  

This dual form of the Saturation Set will turn out to be useful for early duplicate removal and will provide a series of advantages pointed out throughout this section. We stress out however the fact that the explicit maintaining of atoms can be altogether skipped and duplicate removal delayed to the the target instance.\\

\\

\noindent \textbf{Chase steps}. The chase steps are implemented in the $AddAssignmentAndChase$ routine, taking as argument an assignment $a$ and the current Saturation Set $\Ssat$. The $AddAssignmentAndChase$ routine will :
\begin{itemize} 
\item fire a tgd step with the assignment $a$ 
\item consequently, iteratively fire all applicable egds until no more egds apply.
\end{itemize}

The pseudocode for the $AddAssignmentAndChase$ is presented in figure ...

The routine starts by \\

\noindent \textbf{Building and processing Saturation Sets}. The core of our algorithm resides in the construction and chasing of Saturation Sets, taking place at lines \ref{startbuildsat}-\ref{endbuildsat}. The procedure starts with a {\em seed assignment} which is an assignment not attributed to a previously constructed Saturation Set (line \ref{chooseseed}. It continues by alternating between (i) reaching out to new assignments that overlap with existing ones (line \ref{search}) and (ii) adding these new assignments in the Saturation Set and chasing the Saturation Set to termination using the $AddAssignmentAndChase$ routine.

An important question that may arise at this point is: why interleave assignment discovery and chase? Indeed, while adding an assignment to a Saturation Set can be seen as a {\em chase step with a tgd} (we remind that this consists in picking and consuming an assignment), egd application could in principle be postponed to the end of the Saturation Set construction. 
The essential observation here is that {\em early egd application can help diminish the size of the produced Saturation Sets}. This has been hinted at the end of Subsection..., and will be further detailed in Subsection... 
\\}

\subsection{Optimizing assignment discovery}

The loop on line \ref{expandsat} in the $BuildAndChaseSaturationSet$ subroutine inspects all not-yet-consumed assignments for overlaps with the currently considered (frontier) assignment $a$. The overlap check for $a$ and $a'$ furthermore implies, in principle, examining all fds and all term pairs for every pair of assignments, thus implying an high overhead. We will show how to importantly lower this overhead, by leveraging {\em relations between the respective tgds}. 

We start by defining the notions of tgd {\em conflict areas} and {\em conflicts}, as follows:

\begin{definition}[Conflict area]
Let $m$ be an s-t tgd and $f$ an fd of the form $R.A \rightarrow R.B$.
Let $ct$ = $(v_1, \dots, v_n)$ be a tuple of $n$ terms occurring in the head of $m$, where $n$ is the cardinality of the set $A$.
Then the pair <$ct,f$> is a conflict area for $m$ iff: 
\begin{itemize}
\item there is an $R$ atom in the head of $m$ comprising each $v_i$ at the $i$th position of $A$ and
\item all $v_i$ are either universal variables, or constants, or mutable existentials.
\end{itemize} 
\end{definition}

\begin{definition}[Conflicts and conflicting tgds]
Given\\ two (non necessarily distinct) s-t tgds $m_1$ and $m_2$ and conflict areas $ca_1 = <ct_1, f_1>$ and $ca_2 = <ct_2, f_2>$ for $m_1$ and $m_2$ respectively, the pair $<ca_1, ca_2>$ is a conflict between $m_1$ and $m_2$ iff $f_1=f_2$.

If $m_1$ = $m_2$ and $ca_1 = ca_2$ we further say that the conflict is {\em trivial}.
\end{definition}

We characterize all conflicts among the s-t tgds in a DE scenario by means of the {\em Conflict Graph}:

\begin{definition}[Conflict Graph]
Given a Data Exchange scenario, its associated {\em Conflict Graph}, denoted by \textit{$G^c_{\Sigma_{st}}$}, is an undirected (possibly) cyclic graph
composed by an ordered pair $(V,E)$, where $V$ is a set of vertices, each
of which represents an s-t tgd in $\Sigma_{st}$, and $E$ is a set of edges, each edge witnessing the presence of at least one non-trivial conflict between the two
s-t tgds corresponding to its incident vertices. Furthermore, every vertex in $V$ is adorned with all the conflict areas for the corresponding tgd. 
\end{definition}  

\begin{example}
The Conflict Graph for the DE scenario of Example \ref{ex:running-example} is depicted below:
\end{example}

\begin{figure}[h!]
	\centering

\begin{tikzpicture}[=stealth',auto,node distance=3cm,
thick,main node/.style={circle,fill=blue!20,draw,font=\sffamily\Large\bfseries}, every loop/.style={}]

\node[circle,draw] (1) {$v_1$};
\node[circle,draw] (2) [left of=1] {$v_2$};
\node[circle,draw] (3) [right of=1] {$v_3$};
\path[]
(1) edge (2)
(1) edge (3)
(2) edge [bend right] (3)
(3) edge [loop above] (3);
\end{tikzpicture}
	$\begin{array}{c}
	v_1, v_2, v_3 \text{ are vertices corresponding to the tgds $m_1$, $m_2$, $m_3$}
	\vspace{1em}\\
	\mathbf{Areas(v_1)}=\{ca_1^1=\langle(n,s),e_1\rangle\}.\\
    \mathbf{Areas(v_2)}=\{ca_2^1=\langle(n',s'),e_1\rangle,ca_2^2=\langle(p',w'),e_2\rangle\}.\\
    \mathbf{Areas(v_3)}=\{ca_3^1=\langle(n'',s''),e_1\rangle,ca_3^2=\langle(n''',s'''),e_1\rangle\}.
    \vspace{1em}\\

	\end{array}$
\end{figure}

As we will show, the Conflict Graph will allow increasing the efficiency of the new assignment discovery. To characterize the link between conflicts and overlaps, we will introduce the notions of {\em conflict masks} and their {\em matching}, as follows:

\begin{definition}[Conflict mask]
Let $m$ be an s-t tgd, $a$ an assignment for $m$ and $ca = <(v_1, \dots, v_n), f>$ a conflict area for $m$.
The conflict mask of $a$ on $ca$, denoted by $mask_a^{ca}$, is a tuple $(w_1,\dots,w_n)$ such that for all $i=1,\dots,n$, $w_i \in \Deltac \cup \{*\}$ and
\begin{itemize}
\item if $a(v_i)$ is a constant then $w_i = a(v_i)$
\item else, $w_i$ = $*$.
\end{itemize}
We will hereafter also refer to conflict masks without detailing their source, that is, considering masks as $n$-tuples comprising as elements either constants or the special $*$ value.
\end{definition}

\begin{definition}[Matching of a conflict mask]
Let $m$ be an s-t tgd, $ca$ = $<(v_1,\dots, v_n), f>$ a conflict area for $m$ and $a$ an assignment for $m$. Let $msk$ be a conflict mask of size $n$. We say that $a$ {\em matches} $msk$ on $ca$ iff, for all $i=1,\dots,n$ where $msk_i$ is a constant, either $a_2(v_i)$ = $msk_i$ or
$a_2(v_i) \in \Deltan$.
\end{definition}

Based on the above definitions, we are now ready to show the essential link between overlaps, conflicts and conflict masks:

\begin{proposition}\label{prop:conflict-overlap}
Let $m_1$ and $m_2$ be two (non-necessarily distinct) s-t tgds and $a_1$ and $a_2$ two distinct assignments of $m_1$ and $m_2$ respectively. Let $f = R.A \rightarrow R.B$ be a functional dependency with $n$ the cardinality of $A$.

Let $V$ = $(v_1,\dots, v_n)$ and $V'$ = $(v'_1,\dots,v'_n)$ be n-tuples of terms occurring in the head of $m_1$ and $m_2$ respectively. Then the following are equivalent:
\begin{itemize}
\item $a_1$ and $a_2$ overlap on $<v_1, v'_1>, \dots <v_n, v'_n>$ and $f$
\item $<ca_1, ca_2>$ with $ca_1$= $<(v_1,\dots, v_n),f>$ and $ca_2$ = $<(v'_1,\dots, v'_n),f>$ is a conflict between $m_1$ and $m_2$ and $a_2$ matches $mask_{a_1}^{ca_1}$ on $ca_2$.
\item $<ca_1, ca_2>$ with $ca_1$= $<(v_1,\dots, v_n),f>$ and $ca_2$ = $<(v'_1,\dots, v'_n),f>$ is a conflict between $m_1$ and $m_2$ and $a_1$ matches $mask_{a_2}^{ca_2}$ on $ca_1$.
\end{itemize}
\end{proposition}

\noindent \textbf {Using the Conflict Graph and conflict masks}. We will use conflict masks and their matching in conjunction with the Conflict Graph to direct and filter new assignment discovery. By its edges, the Conflict Graph will allow us to only consider specific assignments (corresponding to neighbour tgds). Conflict masks derived from the vertex conflict areas will further act as search masks for matching assignments. Intuitively, in such searches, constants will impose the search for matching constants, whereas the wildcard ($*$) character allows any value on the corresponding position.

Once one is sure that all assignments matching a conflict mask for a given fd have been added to a Saturation Set, the  conflict mask can be {\em marked as exhausted} for the fd and never be employed again in future searches. Furthermore, {\em less permissive} masks for the same fd can be discarded too. We formalize this by the concept of {\em mask subsumption}:
\begin{definition}[Mask subsumption]
Let $msk$ and $msk'$ be two conflict masks of the same size. We say that $msk$ {\em subsumes} $msk'$, denoted by $msk \succeq msk'$, iff for all $i$, $msk_i = msk'_i$ or $msk_i$=$*$.
\end{definition}


\noindent \textbf{Putting it all together.}
Based on the concepts defined above we are ready to present the final, refined and optimized version of our algorithm. Changes on the main loop concern solely the construction of the Conflict Graph, prior to all other operations. This implies the addition at the very beginning of our algorithm of the line:\\

$G^c_{\Sigma_{st}}$ = BuildConflictGraph($\Sigma_{st}$);\\

\begin{algorithm}[h!]
\nonl \textbf{Subroutine} BuildAndChaseSaturationSet\\
\nonl {\em\ \ \ Can raise Exception:ChaseFail}\\
\KwIn{$seed$}
\nonl \textbf{Globals: $A_m \forall m \in \Sigmast$, $G^c_{\Sigma_{st}}$}\\
\KwOut{Saturation Set $\Ssat$}

$NewAssignments$ = $\emptyset$\;
$\Ssat = \emptyset$\;
$UsedMsk = \emptyset$\;
AddAndProcessAssignment($seed$, Tgd($seed$))\;
\While{$NewAssignments \neq \emptyset$}
{ 
	$a$ = SelectArbitraryAsg($NewAssignments$)\;\label{picknewfrontasg}
    $NewAssignments - = \{a\}$\;
    \ForEach {$ca=<cv, f>  \in$ Areas($G^c_{\Sigma_{st}}$,Tgd($a$))}
    {\label{startsearch}
    	$msk$ = $mask_a^{ca}$\;\label{buildmask}
    	\If {$\exists <msk',f>\in UsedMsk$ s.t. $msk'\succeq msk$}
    	{
    		$Continue$\; \label{abortsearch}
    	}
    	\ForEach {$a' \in A_{Tgd(a)}$ matching $msk$ on $ca$}
        {\label{checktrivial}
        	AddAndProcessAssignment($a'$)\;
        }
        \ForEach {$m'$ = Neighbour($G^c_{\Sigma_{st}}$, Tgd($a$)) 
        \textbf{and\ } $ca'=<cv',f> \in$ Areas($G^c_{\Sigma_{st}}$,$m'$)}
    	{\label{getneighbours}
    		    \ForEach {$a' \in A_{m'}$ matching $msk$ on $ca'$}
        		{\label{getassignments}
        			AddAndProcessAssignment($a'$)\;
        		}
        }
        $UsedMsk \cup = \{<msk,f>\}$\;
    }\label{endsearch}
    	
}
\Return $\Ssat$\;
\ \ \\
\textbf{macro} $AddAndProcessAssignment(asg)$:\\
\ \ $NewAssignments \cup= \{asg\}$\;
\ \ $\Ssat \cup = \{asg\}$\;
\ \ $A_{Tgd(asg)} -= \{asg\}$\;
\ \ ApplyEgdsToTermination($\Ssat, asg$)\;
\caption{BuildAndChaseSaturationSet, optimized\label{buildsatv1}}
\end{algorithm}

As expected, the bulk of changes concerns the $BuildAnd\-ChaseSaturationSet$ subroutine, whose final, optimized version is depicted by Algorithm \ref{buildsatv1}. For conciseness, we have grouped the sequence of operations systematically applied to all newly added assignments in a macro called $AddAnd\-ProcessAssignment$. Note that these operations do not vary from our initial version of the subroutine. The major changes, in turn, concern the new assignment discovery. We indeed {\em replace and refine} the loop of line \ref{expandsat} of Algorithm \ref{buildsatv0} with lines \ref{startsearch}-\ref{endsearch} of Algorithm \ref{buildsatv1}. We detail below the purpose and content of these changes.

As in our initial version of the subroutine, we deal with the case of expanding the Saturation Set by reaching out to new assignments that overlap with an already present assignment. To this purpose however we now rely on the Conflict Graph and conflict masks. Our search for new assignments overlapping with a given assignment $a$ starts with  by iterating on the conflict areas that annotate the Conflict Graph node corresponding to the s-t tgd of $a$, which we will hereafter denote by $m_a$. This takes place on line \ref{startsearch}.

For each conflict area, we construct the corresponding conflict mask (line \ref{buildmask}). If the constructed mask (or a mask that subsumes it) 
has been previously used in conjunction with the fd corresponding to the area, the search is aborted (line \ref{abortsearch}), as we are sure that it will not yield any new assignments.

Otherwise, the search for new assignments continues by first considering trivial conflicts (line \ref{checktrivial}) on the given conflict area, between $a$ and other assignments of $m_a$. Once trivial conflicts have been examined, non-trivial conflicts are then considered. This in turn will imply the usage of the the Conflict Graph, to detect all conflicts of $m_a$ on $ca$: neighbour s-t tgds $m'$ conflicting with $m_a$ on $ca$ are identified and specific conflict areas $ca'$ determined (line \ref{getneighbours}).

Note that for both trivial conflicts and other, Conflict Graph edges-based conflicts, we use conflict areas and the constructed conflict mask to evaluate the overlap as a matching of a mask (lines \ref{checktrivial} and \ref{getassignments}). This indeed corresponds to our previous alternative characterization (by Proposition \ref{prop:conflict-overlap}) of overlaps. 

\subsection{Parallelization}
Before stating the correctness of our algorithm and proceeding to its practical evaluation, we note a last optimization opportunity, namely a {\em parallelization} opportunity. 

Since Saturation Sets are discovered and chased in the same time,
we cannot directly envision their parallel processing. 
However, it turns out that there exist, for a given assignment set, {\em groups of Saturation Sets} that can be treated in parallel. Such groups are induced by the connected components of the Conflict Graph.

Indeed, note that assignments for s-t tgds belonging to two distinct connected components in the Conflict Graph will never overlap (since overlap would imply conflict). Since our Saturation Set discovery is based on overlap, this in turns implies that for any two assignments placed in the same Saturation Set by our $BuildAndChaseSaturationSet$ subroutine, their corresponding s-t tgd nodes belong to the same Conflict Graph connected component. 

Assuming a call $Tgds(c)$ that returns all s-t tgds whose corresponding nodes belong to a connected component $c$ of our conflict graph $G^c_{\Sigma_{st}}$, we can thus modify lines \ref{bigloop}--\ref{endsatdesc} in Algorithm \ref{mainloopv0} to
\begin{algorithm}[h]
\ForEach{Connected Component $c$ of $G^o_{\Sigma_{st}}$ {\bf in parallel}}
{
\While{$\exists m \in Tgds(c)$ s.t. $A_m \neq \emptyset$}
{
	$seed$ = SelectArbitraryAsg($A_m$)\;
	$\Ssat$ = BuildAndChaseSaturationSet($seed$)\;
    //\textit{add result to the target instance}\\		 
    $J = J \cup$ Materialize($\Ssat$)\;
}
}
\end{algorithm}

We will call this version of our algorithm the {\em \ouralgo\ With
Parallelization} and retain it for the experimental evaluation hereafter.
We further show that the following holds\footnote{We refer the reader to the appendix for a full proof of this result.}:

\begin{theorem}\label{TH:COR}
If the Oblivious Chase fails \footnote{\label{note:iso}In virtue of
Proposition \ref{prop:iso}, we use a generic term for "Oblivious Chase
failure" or "Oblivious Chase result" regardless of the specific sequence}
then both the \ouralgo\ and the \ouralgo\ With Parallelization fail.
Else, any solution issued by running the \ouralgo\ or the \ouralgo\ With
Parallelization is isomorphic to the result of the Oblivious Chase.
\end{theorem}

\begin{proof}[Sketch] 
We rely on two essential observations: (i) disjoint Saturation Sets completely partitioning the assignment set can be {\em chased individually} and the results {\em pieced together} to obtain a solution, and (ii) the outputs of our
$BuildAndChaseSaturationSet$ subroutine are disjoint Saturation Sets, chased to termination, and completely partitioning the assignment set. To show that each subroutine call returns a Saturation Set, we will in turn rely on the relations between conflicts and overlaps, and overlaps and
collisions respectively, thus on Propositions \ref{prop:conflict-overlap}
and \ref{prop:overlap_collision}. 
\end{proof}

\section{Experimental study}

We have implemented the Interleaved Chase With Parallelization in
a DE engine written in Java and dispatching intermediate SQL calls. We dedicate this section to studying our algorithm's practical behaviour. We start by assessing its standalone performance w.r.t. parameters such as the size of the source instances and the number of s-t tgds and target fds. We continue by a comparison with state-of-the-art DE engines. All our experiments were run on an 8-cores, 3.50 Ghz processor and 28GB RAM computer under Linux Debian. We used Postgres 9.4 as an underlying DBMS for both our system and the DE engines we compare to.\\

\begin{figure}[t!]
	\centering
	 \scalebox{0.7}{\begin{tabular}[h!]{|c|c|c|}
			\multicolumn{3}{l}{}\\
			\multicolumn{3}{l}{\textbf{Scenarios}}\\
			\hline {\it SCENARIO}  & {\it s-t tgds} &  {\it egds} \\ \hline
			OF & 30 & 10 egds \\ \hline
			OF+ & 30 & 20 egds \\ \hline
			OF++ & 30 & 30 egds \\ \hline	
		\end{tabular}}
	
	\centering
	\subfigure[\scriptsize Varying the size of the source instance.]{
		\includegraphics[width=6.5cm]{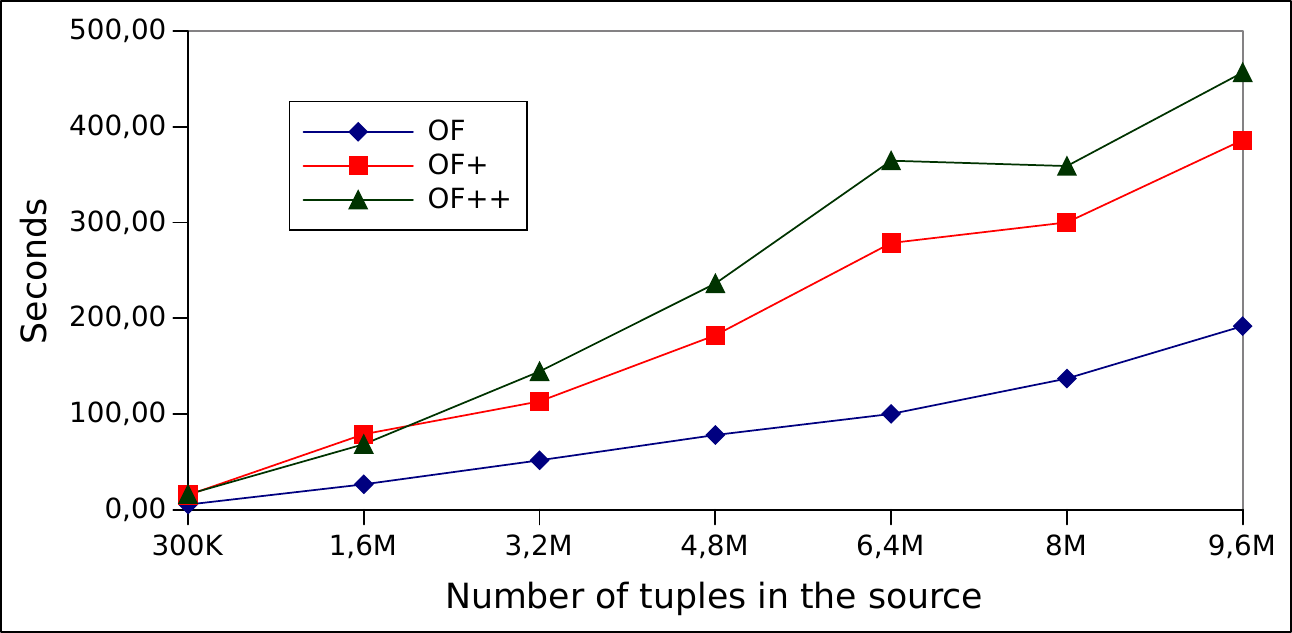}
	}
       \centering
       \scalebox{0.7}{\begin{tabular}[h!]{|c|c|c|c|c|c|}
			\multicolumn{5}{l}{}\\
			\multicolumn{5}{l}{\textbf{Scenarios}}\\
			\hline {\it SCENARIO}  & {\it s-t tgds}  & {\it OF}  & {\it OF+} & {\it OF++}  &  {\it \# source tuples} \\ \hline
			A & 3 & 1 egds & 2 egds & 3 egds & 450K \\ \hline
			B & 4 & 1 egds & 2 egds & 3 egds & 600K \\ \hline
			C & 5 & 1 egds & 2 egds & 3 egds & 750K \\ \hline
		 	D & 6 & 1 egds & 2 egds & 3 egds & 900K \\ \hline
		 	E & 7 & 1 egds & 2 egds & 3 egds & 1.05M \\ \hline
			
		\end{tabular}}

\centering
	\subfigure[\scriptsize Varying more parameters, no parallelization.]{
	\includegraphics[width=6.5cm]{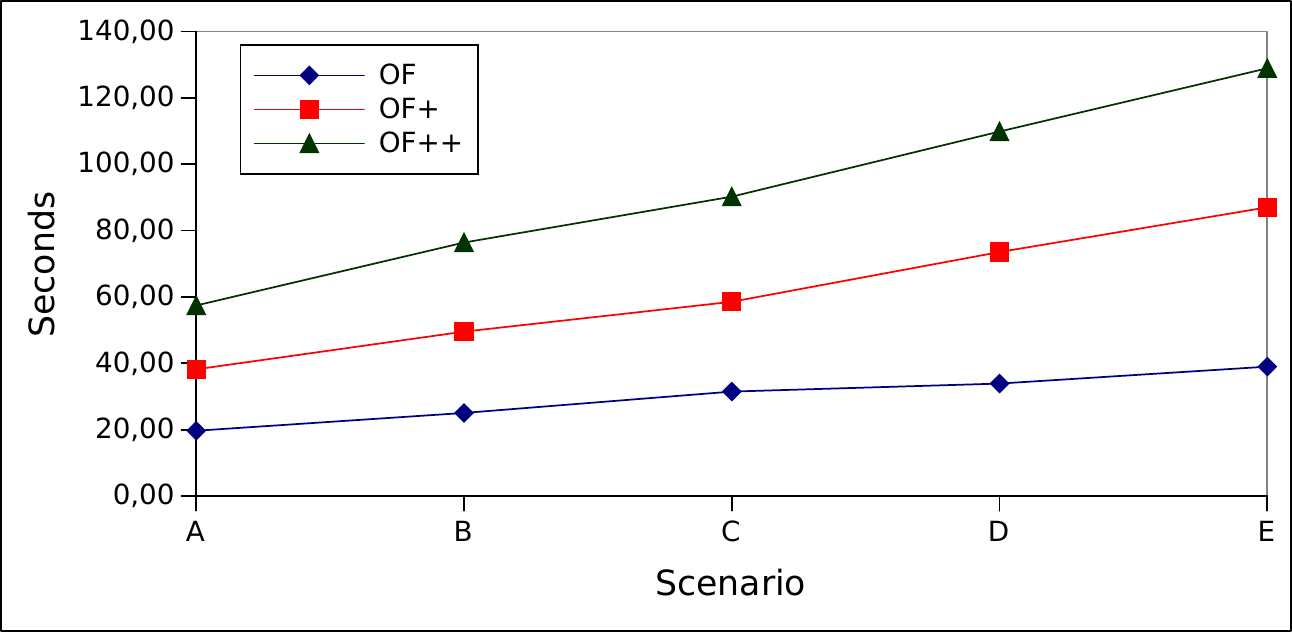}
	}
	
	\centering
	\scalebox{0.7}{\begin{tabular}[h!]{|c|c|c|c|c|c|}
			\multicolumn{5}{l}{}\\
			\multicolumn{5}{l}{\textbf{Scenarios}}\\
			\hline {\it SCENARIO}  & {\it s-t tgds}  & {\it OF}  & {\it OF+} & {\it OF++}  &  {\it \# source tuples} \\ \hline
			A & 15 & 5 egds & 10 egds & 15 egds & 500K \\ \hline
			B & 30 & 10 egds & 20 egds & 30 egds & 1M \\ \hline
			C & 45 & 15 egds & 30 egds & 45 egds & 1.5M \\ \hline
			D & 60 & 20 egds & 40 egds & 60 egds & 2M \\ \hline
			E & 75 & 25 egds & 50 egds & 75 egds & 2.5M \\ \hline
			F & 90 & 30 egds & 60 egds & 90 egds & 3M \\ \hline
			
		\end{tabular}}
		
	\centering
	\subfigure[\scriptsize Varying more parameters, with parallelization.]{
	\includegraphics[width=6.5cm]{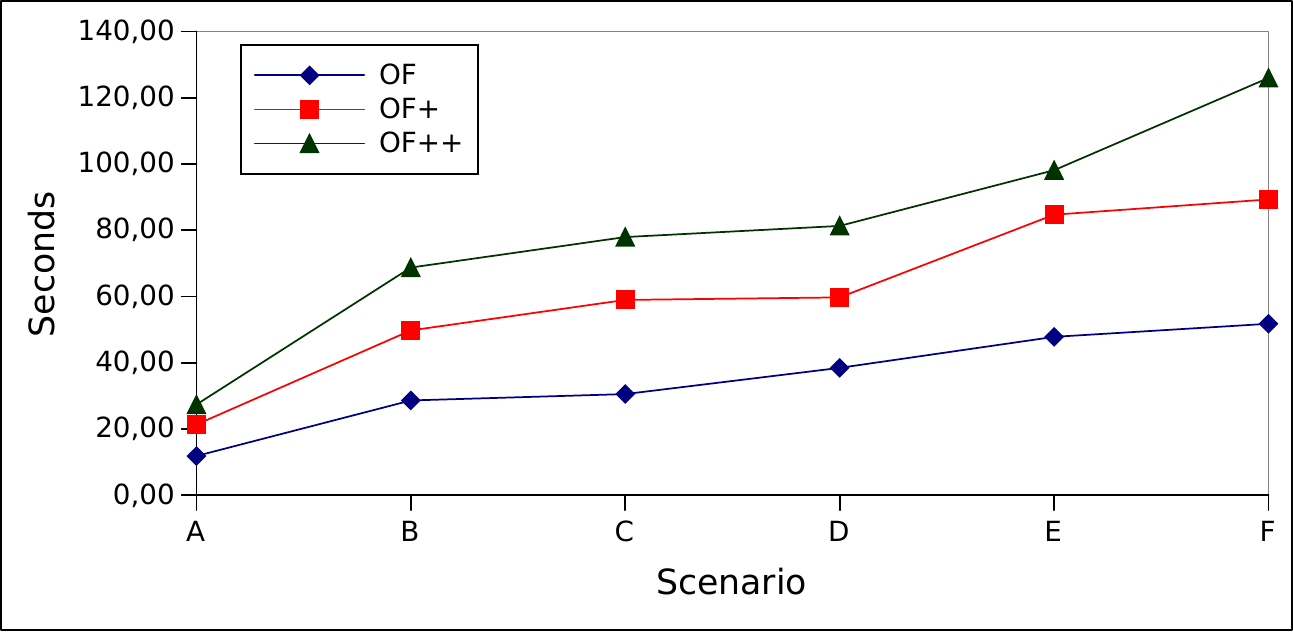}
	}
	
	\caption{\label{fig:exp-red}Scalability of our algorithm w.r.t. various parameters and benefits of parallelization.}
\end{figure}

\noindent \textbf{Benchmarking scenarios}. To evaluate our algorithm, we have used scenarios generated by using iBench \cite{ACGM15}, a recent data integration benchmark for arbitrarily large and complex schema mappings, schemas and schema constraints. 
We have considered three sets of scenarios: (i) $\mathsf{OF}$ scenarios generated with the default iBench {\em
object fusion} primitive; (ii) $\mathsf{OF^{+}}$ scenarios,
generated by combining the {\em object fusion} and {\em vertical partitioning} iBench primitives, leading to s-t tgds with two atoms in the head; (iii)
$\mathsf{OF^{++}}$ scenarios, obtained via a newly created iBench primitive
that provides settings similar to $\mathsf{OF^{+}}$ but with s-t tgds
of increasing complexity (i.e. having three atoms in the head). 
Our iBench-based experimental scenarios are available at \cite{michelespage}.
\\

\noindent {\bf Standalone performance and scalability}. In our first experiment, we have fixed the number of s-t tgds and egds for each type of scenario (table of Figure \ref{fig:exp-red}(a)), and we varied
the number of tuples in the source instances from $300K$ (thousands) to $9.6M$ (millions). 
The goal of this experiment was to measure our algorithm's performance for source instances of increasing size. 
Figure \ref{fig:exp-red}(a) shows how this time quasi-linearly scales w.r.t. the number of source tuples. 

We further tested our algorithm's scalability w.r.t. more parameters, including not only the source instance size but also the number of dependencies, as shown in the tables of Figure \ref{fig:exp-red}(b) and (c). 
These further experiments are also targeted towards the analysis of the benefits of paralellization.

Figure \ref{fig:exp-red}(b) shows the solution generation time for scenarios in which we gradually increase the size of the Conflict Graph (by adding more nodes as s-t tgds and more edges based on egds), while ensuring this graph consists of only one connected component. 
Since the Conflict Graph is completely connected we do not exploit any parallel computation. 
As such, this experiment can be seen a test for the Interleaved Chase without parallelization. 
The trend in Figure \ref{fig:exp-red}(b) shows that our algorithm does not remarkably suffer from adding nodes and edges to the Conflict Graph and that the overhead is quite sustainable. 

We then consider (Figure \ref{fig:exp-red}(c)) Conflict Graphs comprising several connected components, in order to study the impact of parallel evaluation. 
We start from the previous scenarios and progressively add batches of $15$ s-t tgds, yielding $15$ additional nodes and $5$ additional connected components. 
This experiment shows how leveraging the Conflict Graph's structure and parallel evaluation allows mitigating the overhead brought by the increasing number of tuples in the source and the increasing number of dependencies. 
Indeed, the reduced egd application scope, paired with the parallelization on connected components of the Conflict Graph, allows us to stay within the same order of magnitude as previously (Figure \ref{fig:exp-red} (b)). 
This is, on our opinion, an interesting and promising result.
{\em Indeed, to the best of our knowledge, we are not aware of other chase engines capable of evaluating these high numbers of constraints in Data Exchange}. \\

\noindent {\bf Comparison with state-of-the-art DE engines}.
We further set to compare our algorithm with the latest DE engines publicly available, more precisely with an all-SQL DE engine, ++Spicy \cite{MarnetteMPRS11}, and a Java-based custom chase engine, namely Llunatic \cite{GeertsMPS14} \footnote{We have used ++Spicy Version 1.1 (July 2015) \cite{spicy} and 
Llunatic Version 1.0.2 (Nov. 20, 2015) \cite{llunatic}, respectively. We omit the comparison with another Java-based custom chase engine \cite{PichlerS10}, since Llunatic already outperforms it.}. 
While the former is a pure mapping chase engine, the latter is a mapping and cleaning chase engine, out of which we only use the mapping (DE) part (i.e. we do not leverage cleaning constraints). 
Moreover, in order to ensure a fair comparison with ++Spicy, we switch off its core computation. 

When running our experiments, we have set a timeout threshold at 600 seconds (10 minutes), considering that beyond this threshold the practical appeal of achieving solution computation decreases. 
Figure \ref{fig:Comparison}(a) shows the results in terms of solution generation times for the three systems. 
One can note that our algorithm exhibits performance similar to ++Spicy's up to a cut-off point (1.6M tuples); then, {\em by contradicting the common wisdom that custom chase engines are less efficient than all-SQL ones}, our system exhibits better performance for large sizes of the source instance. 
Llunatic's execution on the other hand unfortunately reached the timeout threshold on all tested scenarios, without outputting a solution. 
We recall however that Llunatic is a system that achieves mapping and cleaning at the same time (which we do not). 
Presumably thus, Llunatic is not strongly optimized for a "mapping standalone" mode in the DE scenarios under test. 

Importantly, note that in this comparative experiment we have focused on the $\mathsf{OF}$ scenarios, i.e. the baseline scenarios considered in our study (with parameters in the first table of Figure \ref{fig:exp-red}). 
We did so because more complex scenarios such as $\mathsf{OF^{+}}$ and $\mathsf{OF^{++}}$ {\em cannot be handled by ++Spicy} without additional source dependencies being explicitly provided. 
$\mathsf{OF^+}$ and $\mathsf{OF^{++}}$ scenarios on the other hand {\em can} be handled by Llunatic, but the resulting run meets the timeout threshold in a similar manner as for $\mathsf{OF}$. To mitigate the arguably subjective choice of the timeout threshold (and its effect on the Llunatic measured times), as well as to further test and validate our algorithm's performance, we ran an additional comparative experiment involving Llunatic. This time, our experiment uses the largest scenarios provided with the Llunatic chase engine \cite{llunatic}, adapted to the settings that both our engine and Llunatic can handle, i.e. s-t tgds and target fds. Figure \ref{fig:Comparison} (b) shows the results of this additional comparison, which again turn out to be advantageous for our algorithm. \\

\noindent \textbf{Discussion}. Summarizing, our algorithm gracefully scales thanks to diminishing the egd overhead by reducing the scope of application thereof. Furthermore, the parallel execution leads to scalability in increasingly complex scenarios with high numbers of constraints. This allows us to stand the comparison not only with custom chase engines, but also with all-SQL DE engines on the classes of scenarios for which the latter have improved performances. 

\begin{figure}[t!]
	\centering	
	\subfigure[\scriptsize Our algorithm versus Llunatic and ++Spicy on $\mathsf{OF}$ scenarios.]{
		\includegraphics[width=6cm]{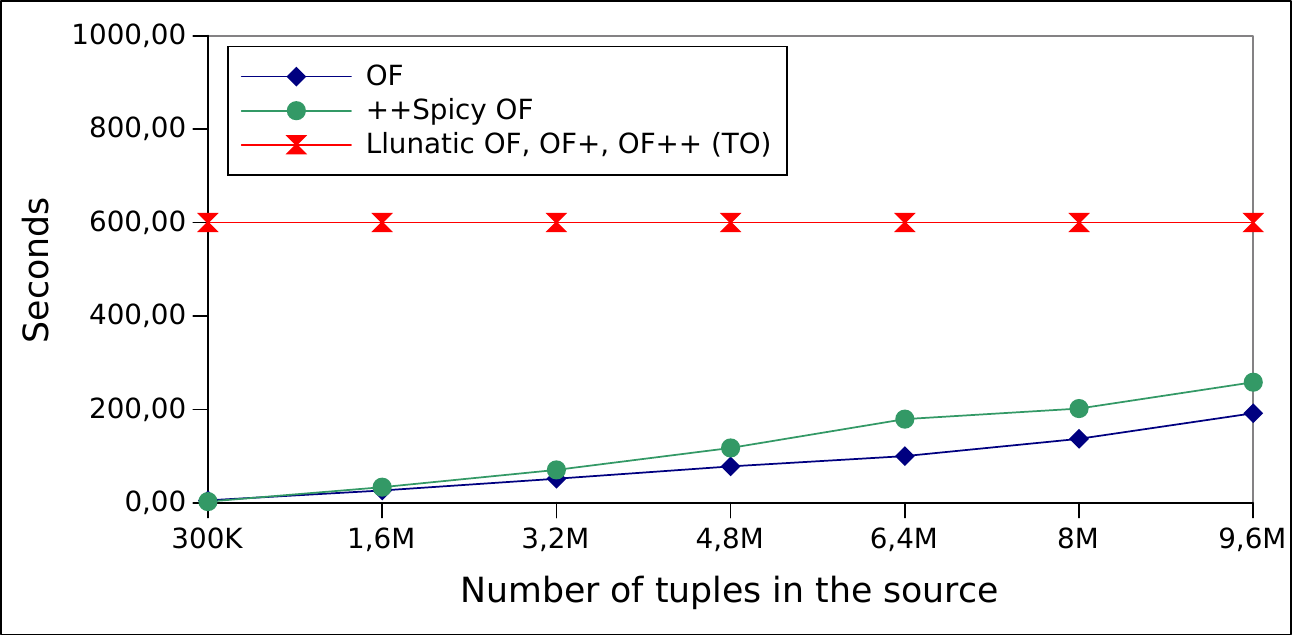}
	}
	
	\scalebox{0.7}{\begin{tabular}[h!]{|c|c|c|c|}
							\multicolumn{4}{l}{}\\
					\multicolumn{4}{l}{\textbf{Scenarios}}\\
							\hline {\it Scenario}  & {\it s-t tgds}  & {\it egds} & {\it \# source tuples} \\ \hline
							Employees & 8 & 2 egds & 500K \\ \hline
							Employees & 8 & 2 egds & 1M \\ \hline
							Workers & 7 & 2 egds & 500K \\ \hline
						\end{tabular}}

	\subfigure[\scriptsize Our algorithm versus Llunatic on Llunatic's scenarios.]{
		\includegraphics[width=5.5cm]{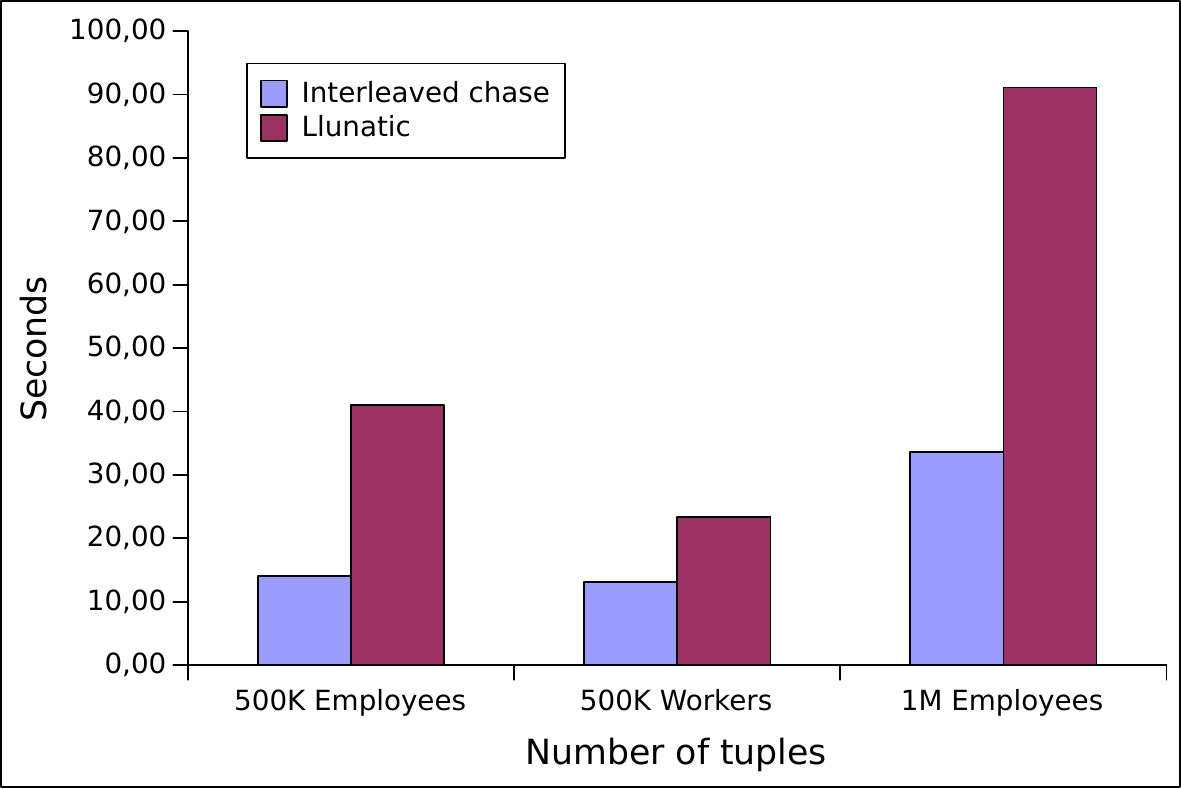}
	}

	\caption{\label{fig:Comparison}Comparison between our algorithm and state-of-the-art DE engines.}\hspace*{-10cm}
\end{figure}

\section{Related work}
Equality-generating dependencies have been largely disregarded as a class of target constraints in Data Exchange settings and have started to be taken into consideration only recently \cite{GottlobPS11,MarnetteMP10}. 
In \cite{GottlobPS11}, normalization and optimization techniques are applied to s-t tgds, and extended under particular assumptions to target egds. 
One such assumption is the presence of source constraints, which is also adopted in \cite{MarnetteMP10} to obtain a first-order rewriting of a set of s-t tgds along with a set of egds, in order to leverage the efficiency of pure SQL. 
Our approach is more general with respect to the one presented in \cite{MarnetteMP10}, as it also covers sets of constraints that are not first-order rewritable. 
On the other hand, the algorithm in \cite{MarnetteMP10} features core computation, which we do not provide.  
We were able to compare our algorithm with ++Spicy \cite{MarnetteMP10} only for scenarios that are 
FO-rewritable, as it has been presented in Section 5.
Despite the fact that ++Spicy is based on pure SQL, we have shown that our custom chase exhibits better performance on large instances in our tested scenarios. 
The Llunatic mapping and cleaning system \cite{GeertsMPS14} is tackling the problem of mapping and cleaning in the same chase step. 
The Llunatic chase efficiently covers a range of scenarios, such as repairing source data by also exploiting user feedback, that are not covered in our algorithm. 
Nevertheless, since Llunatic's mapping engine can also be used in a standalone mode and Llunatic has been shown to exhibit good DE performance, we have presented  comparative experimental results with Llunatic in Section 5. These experiments show the superiority of our algorithm on the tested DE scenarios, by up to several orders of magnitude.  

We propose the Interleaved Chase, a chase-based algorithm that interleaves target fd and s-t tgd evaluation, aiming for the construction of smaller intermediate chase results. 
Our algorithm stands in between the Classical DE Chase, where the entire target instance is constructed 
prior to applying egds, and the general Oblivious variant proposed by \cite{GrahneO14} which, differently from the 
Oblivious Chase in \cite{CaliGK13}, does not rely on any order of application of tgds and egds.
We recall also that the Oblivious Chase is the simplest chase variant to implement in DE, and is in fact underlying the Classical DE Chase inside the two halves of s-t tgd and egd evaluation.
In \cite{CaliGK13}, an explicit ordering of the Oblivious Chase is introduced to detect innocuous egds, i.e. egds that do not affect tgds and can be ignored in subsequent chase steps. 
The notion of innocuous egds is complementary to the notions presented in this paper and it would be particularly interesting to see how the two could be combined for further improvement. 
Moreover, we do not make any assumption on the fragment of tgds that we can handle, and it would be useful to see how our results could be applied and/or refined for the weakly guarded tgds considered in \cite{CaliGK13}. Last but not least, our static analysis guided by the Conflict Graph lets us obtain more aggressive optimizations, that turn out to be useful in practice.  

Very recently, there has been renewed interest in optimization issues around query answering. 
\cite{CateHK16} presents a practical study of query answering under inconsistency-tolerant semantics. Such study revolves around specific exchange-repair semantics and focuses on the performance of query answering, while our work 
focuses on the efficient interplay of tdgs and egds during the chase for materializing universal solutions, which are meaningful under the classical certain answers semantics of DE. It would be interesting to see how the two approaches can be combined. 
As shown in \cite{georgPods14}, guarded existential rules, that ensure the decidability of query answering tasks, can be translated into plain Datalog rules, in order to 
possibly leverage Datalog optimization techniques. 
However, Datalog rewriting is orthogonal to our approach and a comparison of the performances of the two complementary approaches falls beyond the scope of our work and is left for future investigation. 


The chase procedure is applicable beyond DE, in particular to query reformulation under constraints \cite{DeutschPT99, KonstantinidisA14, BenediktLT15}, including query rewriting with views and access patterns.
The attention devoted to handling egds in these reformulation settings remains however quite limited as well. 
\cite{KonstantinidisA14} presents a chase variant called the Frugal Chase, applicable to both query rewriting and DE, which has the interesting property of yielding smaller universal solutions, even though egds are not covered. 
\cite{BenediktLT15} efficiently deals with query reformulation with constraints and access patterns by exploiting chase-based proofs. The results are however restricted to tgds only. 
\cite{DeutschPT99} and follow-up works devise and refine the versatile Chase\&Backchase algorithm for complete minimal query reformulations, under constraints including general egds. 
However, no specific attention is paid to the task of improving egd
evaluation, and we argue that egd-aimed optimizations, in the spirit of
those that we consider here, could be useful to further boost the performance of the Chase\&Backchase and its variants. 
As it turns out moreover, our algorithm can be directly applicable to query rewriting with views, general query reformulation or other settings implying the usage of the chase procedure, as soon as these settings feature constraints of the type and structure addressed in this paper, i.e. that can be assimilated, via a schema partition, to s-t tgds and target fds.

\balance

\section{Conclusion and Future Work}
In this paper, we have presented a new Data Exchange chase-based algorithm, aimed at efficiently handling general target
functional dependencies. 
We have proved the correctness of our algorithm and shown its performance in an extensive experimental study.
We envision several extensions of our work: the first revolves around the treatment of general target egds, beyond functional dependencies; the second would
lead us to generate (non-standard) recursive SQL queries to implement our
chase with egds; 
the third 
concerns the efficient generation of core solutions \cite{FaginKP05, GottlobN08, CateCKT09}, which is
currently not covered in our algorithm. 
We finally mention a direct, seamless extension of our approach: due to our relying on a DBMS for computing initial s-t tgds assignments, our algorithm can in reality handle larger fragments of s-t tgds, far beyond their conjunctive counterparts and up to the full FO power.\\

\ignore{

\makeatletter
\renewenvironment{thebibliography}[1]{%
\ifnum\addauflag=0\addauthorsection\global\addauflag=1\fi
     \section[References]{
        {References} 
          {\vskip -5pt plus 1pt} 
         \@mkboth{{\refname}}{{\refname}}%
     }%
     \list{[\arabic{enumi}]}{%
         \settowidth\labelwidth{[#1]}%
         \leftmargin\labelwidth
         \advance\leftmargin\labelsep
         \advance\leftmargin\bibindent
         \parsep=0pt\itemsep=1pt 
         \itemindent -\bibindent
         \listparindent \itemindent
         \usecounter{enumi}
     }%
     \let\newblock\@empty
     \raggedright 
     \sloppy
     \sfcode`\.=1000\relax
}

\makeatother
}
\bibliographystyle{abbrv}

{\small
\begin{thebibliography}{10}

\bibitem{ACGM15}
P.~C. Arocena, R.~Ciucanu, B.~Glavic, and R.~J. Miller.
\newblock Gain control over your integration evaluations.
\newblock {\em {PVLDB}}, 8(12):1960--1963, 2015.

\bibitem{BenediktLT15}
M.~Benedikt, J.~Leblay, and E.~Tsamoura.
\newblock Querying with access patterns and integrity constraints.
\newblock {\em {PVLDB}}, 8(6):690--701, 2015.

\bibitem{CaliGK13}
A.~Cal{\`{\i}}, G.~Gottlob, and M.~Kifer.
\newblock Taming the infinite chase: Query answering under expressive
  relational constraints.
\newblock {\em J. Artif. Intell. Res. {(JAIR)}}, 48:115--174, 2013.

\bibitem{DeutschPT99}
A.~Deutsch, L.~Popa, and V.~Tannen.
\newblock Physical data independence, constraints, and optimization with
  universal plans.
\newblock In {\em VLDB'99, Proceedings of 25th International Conference on Very
  Large Data Bases, September 7-10, 1999, Edinburgh, Scotland, {UK}}, pages
  459--470, 1999.

\bibitem{Fagin200589}
R.~Fagin, P.~G. Kolaitis, R.~J. Miller, and L.~Popa.
\newblock Data exchange: semantics and query answering.
\newblock {\em Theoretical Computer Science}, 336(1):89 -- 124, 2005.

\bibitem{FaginKP05}
R.~Fagin, P.~G. Kolaitis, and L.~Popa.
\newblock Data exchange: getting to the core.
\newblock {\em {ACM} Trans. Database Syst.}, 30(1):174--210, 2005.

\bibitem{GeertsMPS14}
F.~Geerts, G.~Mecca, P.~Papotti, and D.~Santoro.
\newblock Mapping and cleaning.
\newblock In {\em {IEEE} 30th International Conference on Data Engineering,
  Chicago, {ICDE} 2014, IL, USA, March 31 - April 4, 2014}, pages 232--243,
  2014.

\bibitem{GottlobN08}
G.~Gottlob and A.~Nash.
\newblock Efficient core computation in data exchange.
\newblock {\em J. {ACM}}, 55(2), 2008.

\bibitem{GottlobPS11}
G.~Gottlob, R.~Pichler, and V.~Savenkov.
\newblock Normalization and optimization of schema mappings.
\newblock {\em {VLDB} J.}, 20(2):277--302, 2011.

\bibitem{georgPods14}
G.~Gottlob, S.~Rudolph, and M.~Simkus.
\newblock Expressiveness of guarded existential rule languages.
\newblock In {\em Proceedings of PODS'14}, pages 27--38, 2014.

\bibitem{GrahneO14}
G.~Grahne and A.~Onet.
\newblock The data-exchange chase under the microscope.
\newblock {\em CoRR}, abs/1407.2279, 2014.

\bibitem{GreenKIT07}
T.~J. Green, G.~Karvounarakis, Z.~G. Ives, and V.~Tannen.
\newblock Update exchange with mappings and provenance.
\newblock In {\em Proceedings of VLDB}, pages 675--686, 2007.

\bibitem{KonstantinidisA14}
G.~Konstantinidis and J.~L. Ambite.
\newblock Optimizing the chase: Scalable data integration under constraints.
\newblock {\em {PVLDB}}, 7(14):1869--1880, 2014.

\bibitem{MarnetteMP10}
B.~Marnette, G.~Mecca, and P.~Papotti.
\newblock Scalable data exchange with functional dependencies.
\newblock {\em {PVLDB}}, 3(1):105--116, 2010.

\bibitem{MarnetteMPRS11}
B.~Marnette, G.~Mecca, P.~Papotti, S.~Raunich, and D.~Santoro.
\newblock {++Spicy: an OpenSource Tool for Second-Generation Schema Mapping and
  Data Exchange}.
\newblock {\em {PVLDB}}, 4(12):1438--1441, 2011.

\bibitem{PichlerS10}
R.~Pichler and V.~Savenkov.
\newblock Towards practical feasibility of core computation in data exchange.
\newblock {\em Theor. Comput. Sci.}, 411(7-9):935--957, 2010.

\bibitem{Clio}
L.~Popa, Y.~Velegrakis, M.~A. Hern\'{a}ndez, R.~J. Miller, and R.~Fagin.
\newblock Translating web data.
\newblock In {\em Proceedings of the 28th International Conference on Very
  Large Data Bases}, pages 598--609, 2002.

\bibitem{CateCKT09}
B.~ten Cate, L.~Chiticariu, P.~G. Kolaitis, and W.~C. Tan.
\newblock Laconic schema mappings: Computing the core with {SQL} queries.
\newblock {\em {PVLDB}}, 2(1):1006--1017, 2009.

\bibitem{CateHK16}
B.~ten Cate, R.~L. Halpert, and P.~G. Kolaitis.
\newblock Practical query answering in data exchange under
  inconsistency-tolerant semantics.
\newblock In {\em Proceedings of {EDBT}}, pages 233--244, 2016.

\bibitem{llunatic}
\url{http://www.db.unibas.it/projects/llunatic/}{~}.

\bibitem{spicy}
\url{http://www.db.unibas.it/projects/spicy/}{~}.

\bibitem{michelespage}
\url{www.mi.parisdescartes.fr/}\midtilde\url{mlinardi/FD_DE_paper.html}.

\end{thebibliography}

}

\pagebreak
\appendix
\section{Proof of theorem \ref{TH:COR}}

\noindent \textbf{Preamble, notation and vocabulary}. Before proceeding to the proof of Theorem \ref{TH:COR}, we start by fixing some notation. Denoting by $A$ the complete assignment set for a given Data Exchange scenario, we define  a finite chase sequence $ChSeq = s_1,\dots, s_n$ on $S \subseteq A$ as a sequence of tgd and egd chase steps $s_i$, all of which are restricted to $S$. We detail hereafter such a sequence.

For each assignment $a$ in $S$ we denote by $a_{init}$ its initial form and by $a_{s_k}$ its form after chase steps $s_1, \dots, s_k$ have been applied. We denote by $T_{s_k}$ the target assignment set in its current form after the chase steps $s_1-s_k$. As explained previously, tgd chase steps simply consist in adding previously non-considered assignments in $S$ to the target assignment set. A tgd chase step $s_i$ with assignment $a$ will thus lead to $T_{s_i} = T_{s_{i-1}} \cup \{a_{init}\}$. Note that any chase sequence has at maximum one step corresponding to an assignment $a$.

Egd steps on the other hand will attempt to replace values in the images of the assignments in the current target assignment set. That is, an egd chase step $s_i$ is a list of replacement attempts of the form: {\em replace $v^j_1$ by $v^j_2$ (where of course $v^j_1<>v^j_2$) in all the images of the assignments in $T_{s_{i-1}}$} . If any of the $v^j_1$ is a constant then the egd chase steps fails. Else the target set $T_{s_i}$ will be the obtained by applying the replacements above on $T_{s_{i-1}}$. 
The {\em result} of a chase sequence after a non-failing step $s_k$, denoted by $J_{s_k}$, is the database instance obtained by head materialization of all the (current form of the) assignments in the target set $T_{s_k}$. 

In the following, we will only consider chase sequences such that $\forall i<n$ the steps $s_i$ are either tgd steps or non-failing egd-steps. We further say that a chase sequence $ChSeq$ on $S$ is terminating iff 

\begin{itemize}
\item either $s_n$ is a failed egd step, in which case we say that $ChSeq$ is a (terminating) failed chase sequence on $S$, or
\item (i) $ChSeq$ comprises exactly one tgd-chase step for each assignment in $S$ and (ii) no egds apply on $T_{s_n}$. In this case we will say that $ChSeq$ is a terminating successful chase sequence on $S$.
\end{itemize}

We further direct our attention to a very important result which, together with Propositions \ref{prop:collision_fixed}, \ref{prop:overlap_collision} and \ref{prop:conflict-overlap}, will turn out essential for our correctness proof:

\begin{proposition}\label{prop:nulls}
Let $A$ be the complete assignment set for a data exchange scenario and let $ChSeq = s_1, \dots, s_n$ be an arbitrary, non-failing but non-necessarily terminating chase sequence on $A$. Then at any moment during $ChSeq$ and for any null $n \in \Deltan$, if $n$ is simultaneously in the image of two distinct assignments $a$ and $a'$, then for every Saturation Set $S$, $a\in S$ iff $a'\in S$.
\end{proposition}
\begin{proof}
We will start by showing that if two assignments have any nulls in common during $ChSeq$, then there exists an assignment sequence $a_1=a, a_2 \dots,a_p=a'$ s.t. all consecutive assignments $a_i$ and $a_{i+1}$ collide. We will hereafter refer to such sequence as a {\em chain of collisions}. 

First, note that prior to $ChSeq$, any distinct assignments $a$ and $a'$ (regardless of their collision) cannot hold any nulls in common. This is due to the construction of $a_{init}$ and $a'_{init}$, by extending body assignments with fresh nulls. 

We then show inductively on the chase sequence $ChSeq = s_1, \dots, s_n$ that the property holds. Assuming that it holds prior to a chase step $s_i$, we show that it holds after step $s_i$. If $s_i$ is a tgd step, then the same reasoning above applies, since all the nulls of the newly introduced assignment (the one picked at step $s_i$) are fresh and do not appear in the target set $T_{s_{i-1}}$.
 
Else, $s_i$ is an egd step. Let $a$ and $a'$ be two distinct assignments that have nulls in common after $s_i$. Two cases may occur:

\begin{itemize}
\item $a$ and $a'$ had nulls in common before $s_i$. Then according to our induction hypothesis there is a chain of collisions linking $a$ and $a'$ (since collision is lifecycle independent, it is not affected by $s_i$).
\item $a$ and $a'$ have nulls in common {\em only after} $s_i$. In this case, let $b$ and $b'$ be (the non-necessarily distinct) assignments whose interaction created the egd-application condition for $s_i$. Let $n_1, \dots, n_z$ be the nulls whose replacement {\em by nulls} is caused by $s_i$ and let $n'_1, \dots, n'_z$ be the nulls after replacement (note that we ignore nulls replaced by constants). Since collision occurs only after $s_i$, it follows that {\em prior to $s_i$} at least one $n_j$ or one $n'_j$ ($1<=j<=z$) is in the image of $a$, and similarly at least one $n_j$ or one $n'_j$ is in the image of $a'$. 

On the other hand, by egd step definition, {\em prior to $s_i$}, all $n_j$ and $n'_j$ are in the image of $b$ or $b'$. It follows that prior to $s_i$, $a$ has nulls in common with $b$ or $b'$, and the same holds for $a'$. Then, by induction hypothesis, there is a chain of collisions between $a$ and $b$, $a_1 = a,\dots\, a_{p_1}=b$, or a chain of collisions $a_1 = a,\dots\, a_{p_1}=b'$ between $a$ and $b'$. The exact same reasoning applies to $a'$, implying either a chain of collisions $a'_1 = a',\dots\, a'_{p_2}=b$ or $a'_1 = a',\dots\, a'_{p_2}=b'$. Note that since collision is lifecycle-independent, these chains of collisions hold both prior to and after $s_i$.

Finally, we note that, since $b$ and $b'$ interact, according to Proposition \ref{prop:collision_fixed}, $b$ and $b'$ collide. We can then produce the chain of collisions $a''_1=a,a''_2 =a_2, \dots, a''_{p_1}=b, a''_{p_1+1}=b', a''_{p_1+2} = a''_{p_2-1}$, $\dots, a''_{p_1+p_2+1}=a'$ (or the equivalent thereof, according to whether the chains corresponding to $a$ and $a'$ link these assignments to $b$ or $b'$). On all cases, it follows that $a$ and $a'$ are linked by a chain of collisions after $s_i$.
\end{itemize}

It is then straightforward to show that the existence of the chain of collisions implies our result. Indeed, wlog, let $S$ be a Saturation Set s.t. $a \in S$. By definition of Saturation Sets, it follows that $a_1 \in S$ (since, due to collision, $a_1$ cannot be in the complement of $S$), then $a_2 \in S$, \dots, then $a_p = a' \in S$.
\end{proof}

Finally, we show below a result useful for reasoning on correctness as well as revealing the practical interest of mask subsumption:
\begin{proposition} \label{prop:masks}
Let $msk_1$ and $msk_2$ be two conflict masks of size $n$ s.t. $msk_1 \succeq msk_2$. Let $a$ be an assignment s.t. $a$ matches $msk_2$ on a conflict area $ca = <(v_1, \dots, v_n), f>$. Then $a$ matches $msk_1$ on $ca$. 
\end{proposition}

\noindent \textbf{Individual chase of Saturation Sets}. Relying on the concepts and results above, we are now ready to start the proof of our correctness theorem \ref{TH:COR}. We will begin by the first main proof element, that is, we will show that Saturation Sets can be chased {\em independently} and the results {\em pieced together} so as to obtain a solution:

\begin{lemma}\label{lemma:point1}
Let $A$ be the complete assignment set for a Data Exchange Scenario and $A_1, \dots, A_m$ be a complete partition of $A$ into $m$ disjoint Saturation Sets, that is $\bigcup_{i=1}^{m}{A_i}=A$ and $\forall i<>j, A_i \cap A_j = \emptyset$.

Let $ChSeq_1 = s^{1}_1,\dots, s^{1}_{n_1}, \dots, ChSeq_m=s^{m}_1, \dots, s^{m}_{n_m}$ be terminating successful chase sequences on $A_1, \dots, A_m$.

Then  there exists a terminating successful chase sequence $ChSeq = s_1,\dots, s_n$ on $A$ such that $J_{s_n} = \bigcup_{i=1}^{m}{J_{s^{i}_{n_i}}}$
\end{lemma}

\begin{proof}
To prove the statement above, we will construct a chase sequence on $A$ by {\em appending} the individual chase sequences on $A_i$. Let thus $ChSeq = s^{1}, \dots, s^{n}$ where:
\begin{itemize}
\item $n=\sum_{i=1}^{m}{n_i}$
\item $s_i$ = $s^{k}_{j}$ where 
\begin{itemize}
\item $k$ is the unique indice of Saturation Set respecting $\sum_{l=1}^{k}{n_l}>=i$ and $\sum_{l=1}^{k-1}{n_l}<i$.
\item $j=i - \sum_{l=1}^{k-1}{n_l}$.
\end{itemize}
\end{itemize}

We will first show that for every step $s_i = s^{k}_{j}$, $T_{s_i} = \bigcup_{l=1}^{k-1}T_{s^{l}_{n_l}}$ $\cup$ $T_{s^{k}_j}$.  We will proceed inductively on the chase sequence $ChSeq$. Two cases can occur:
\begin{itemize} 
\item $s_i$ is a tgd chase step. If $j=1$ then $s_i$ is the beginning of a new append operation, concerning the Saturation Set $A_k$. Denoting by $a$ the assignment picked by this step we have $T_{s^{k}_1}$ = $\{a_{init}\}$. Assuming our inductive hypothesis holds for $s_{i-1}$ this means that $T_{s_{i-1}} =\bigcup_{l=1}^{k-1}{T_{s^{l}_{n_l}}}$. We then have $T_{s_i} =  T_{s_{i-1}} \cup \{a_{init}\}$ = $\bigcup_{l=1}^{k-1}{T_{s^{l}_{n_l}}}$ $\cup$ $T_{s^{k}_j}$. We apply a similar reasoning for $j>1$ by noting that in this case $T_{s^{k}_j}$ = $T_{s^{k}_{j-1}} \cup \{a_{init}\}$.
\item else, $s_i$ must be an egd chase step. Note that in this case $j>1$, since all chase sequences begin with a tgd chase step. By definition of this chase step as being originally operated on $A_k$ and by our inductive hypothesis, for every assignment $a$ in $T_{s^k_j}$, $a_{s_i} = a_{s^{k}_{j}}$. We finish by showing that for all assignments $a'$ that do not belong to $A_k$, $a'_{s_i} = a'_{s_{i-1}}$, in other words, that these assignments {\em are not affected by $s_i$}. 

Indeed, note that since $s_i$ is not a failed chase step, its effect is that of replacing nulls \{$n_1, \dots, n_z$\} (by either nulls or constants). The nulls that will be replaced are in the image of some assignment in $T_{s^k_j}$ (thus in $A_k$) that has created the egd application conditions. Then by Proposition \ref{prop:nulls} it follows directly that these nulls are {\em contained and specific to} the Saturation Set $A_k$, i.e. there exists no assignment $a'$ that (i) does not belong to $A_k$ and that (ii) has in its image one of these nulls. As a consequence, any assignment $a'$ that does not belong to $A_k$ is not affected by the the replacement, and therefore $a'_{s_i} = a'_{s_{i-1}}$.
\end{itemize}

We have thus proved that $T_{s_n} = \bigcup_{l=1}^{m}{T_{s^{l}_{n_l}}}$, which implies that $J_{s_n} = \bigcup_{l=1}^{m}{J_{s^{l}_{n_l}}}$ by definition of head materialization. To complete our proof, we further show that $ChSeq$ is a terminating chase sequence on $A$. Note that, since $A_1,\dots, A_m$ is a complete partitioning of $A$, all tgd steps have been applied. 

We next show that no egd step further applies on $T_{s_n}$. Indeed, assume that an egd chase step still applies. This means that there exist two assignments $a$ and $a'$ s.t. $a_{s_n}$ and $a'_{s_n}$ interact. By Proposition \ref{prop:collision_fixed} it then means that $a$ and $a'$ belong to the same saturation set $S_k$. But this contradicts the fact that all chase sequences $ChSeq_i$ are terminating successful sequences. Therefore, no two assignments $a$ and $a'$ can interact after $ChSeq$, hence $ChSeq$ is terminating, thus completing our proof.

\end{proof}

\noindent \textbf{Output of our algorithms}. We move now to the second main element of our correctness theorem: the fact in case of non-failed runs, the sets of assignments constructed by the $BuildAndChaseSat\-urationSet$ subroutine (in short $BACSS$) of our algorithms (thus, for both the \ouralgo\ and the \ouralgo\ With Parallelization) provide a complete partitioning of the assignment set into disjoint Saturation Sets, and that further each $BACSS$ output is the target set corresponding to a terminating chase sequence on the given constructed Saturation Set. We will indeed show that the following holds:

\begin{lemma}\label{lemma:point2}
Let $A$ be the complete assignment set for a data exchange scenario and 
let $J_{IC/P}$ be the result of a successful run of either the \ouralgo\ or \ouralgo\ With Parallelization for the given scenario.
Then there exist sets $S_1, \dots S_m, S_i \subseteq A$ and terminating successful chase sequences $ChSeq_1 = s^{1}_1,\dots, s^{1}_{n_1}, \dots, ChSeq_m=s^{m}_1, \dots, s^{m}_{n_m}$ on $S_1, \dots, S_m$ such that:
\begin{itemize}
\item $S_i$ are pairwise disjoint and provide a complete partitioning of $A$, that is, $\forall i<>j, S_i \cap S_j = \emptyset$ and $\bigcup_{i=1}^{m}{S_i} = A$
\item $J_{IC/P} = \bigcup_{i=1}^{m}{J_{s^{i}_{n_i}}}$
\item $S_1, \dots, S_m$ are all Saturation Sets.
\end{itemize}
\end{lemma}

\begin{proof}
As sketched above, we are going to take $S_1, \dots$, $S_m$ to be the assignment sets constructed by the calls to $BACSS$ and $ChSeq_1, \dots, ChSeq_m$ to be the chase sequences operated by this subroutine.

First, note that fact that the chase sequences are terminating follows directly from the correctness of the $Apply\-EgdsToTermination$ method and the fact that a tgd step is applied for any assignment added to a constructed set. Also, $J_{IC/P}$ corresponds to the union of the respective chase results by construction, i.e. line \ref{addsat} in Algorithm \ref{mainloopv0} and its equivalent for the main loop of \ouralgo\ With Parallelization. Note further that, upon being picked during some call to $BACSS$, an assignment becomes unavailable for subsequent picking, thus the sets built by $BACSS$ are disjoint. Finally, note that the algorithm(s) will loop until complete exhaustion of the assignment set for the input scenario, thus a complete partitioning is always provided.

We will hereafter focus on showing that calls to $BACSS$ {\em do indeed construct Saturation Sets}. In particular, we are going to show that for every two assignments $a$ and $a'$, if $a$ and $a'$ collide then there exists $S_i$ s.t. $a \in S_i$ and $a' \in S_i$, that is, the two assignments belong to the same constructed set. Since according to the above all assignments belong to one and exactly one constructed set, this in turn will ensure, by Saturation Set definition, that all constructed sets are Saturation Sets.

We start by noting that since $a$ and $a'$ collide, their corresponding rules $m_a$ and $m_{a'}$ {\em belong to the same connected component} of the Conflict Graph. Indeed, by Proposition \ref{prop:overlap_collision}, it follows that $a$ and $a'$ overlap, which by Proposition \ref{prop:conflict-overlap} further implies conflict between the corresponding rules and thus existence of a graph edge. We can thus restrict our attention to strictly one of the parallel branches of the \ouralgo\ With Parallelization and reason in a sequential way for both our algorithms.

Without loss of generality, we can then assume that $a$ is picked {\em} before $a'$, i.e. that when $a$ is picked $a'$ is still available. Let $BACSS_i$ be the call to $BACSS$ that picks $a$ and $S_i$ the set constructed by this call. We will show that $BACSS_i$ {\em also picks $a'$} before exiting, and that therefore $a' \in S_i$. 

We first note that during the execution of $BACSS_i$, $a'$ can only be in one of the following situations: (i) added to $S_i$ (by $BACSS_i$) or still available. This is indeed a consequence of the sequential reasoning shown above, i.e. of the fact that in the case of the \ouralgo\ With Parallelization $a'$ will always be picked on the same branch than the one on which $BACSS_i$ is executing. 

We next consider the point in the execution of $BACSS_i$ where $a$ is selected from the frontier set (recall that this is the $NewAssignments$ set) at line \ref{picknewfrontasg} in Algorithm \ref{buildsatv1}. We will hereafter refer to this point as the {\em expansion} of $a$. Note that such point always occurs, since $a$ has been added to $S_i$ therefore it will be considered for expansion at some point before leaving the frontier set. 

At this point in the execution of $BACSS_i$, two cases may occur: (i) $a'$ has already been added to $S_i$, in which case our proof is completed, or (ii) $a'$ is still available. We thus focus on the second case. 

Since $a$ and $a'$ collide, it follows by Proposition \ref{prop:overlap_collision} that $a$ and $a'$ overlap. Let $<v_1, v'_1>, \dots, <v_n, v'_n>$ be the term pairs and $f$ the functional dependency corresponding to (one of the) overlap(s) of $a$ and $a'$. By Proposition \ref{prop:conflict-overlap} it follows that $m_a$ and $m_{a'}$ conflict on $ca = <(v_1,\dots,v_n), f>$ and $ca'=<(v'_1,\dots,v'_n), f>$ and $a'$ matches $mask_{a}^{ca}$ on $ca'$. By definition of the Conflict Graph, $ca$ then occurs as an adornment of the vertex corresponding to $m_a$. There is then always a run of the loop on lines \ref{startsearch}-\ref{endsearch} for $ca$. Two cases can in principle occur:
\begin{itemize}
\item $mask_{a}^{ca}$ has not been used with $f$. In this case, if $ca=ca'$ and $m_a=m_{a'}$ (thus $<ca, ca'>$ is a trivial conflict), and since $a'$ is still available according to our assumption above, the loop on line \ref{checktrivial} will add $a'$ to $S_i$. Else $<ca, ca'>$ is a non-trivial conflict, therefore by definition it is accounted for by the Conflict Graph, i.e. there is an edge between the nodes corresponding to $m_a$ and.$m_{a'}$ and $ca'$ is among the adornments of the node corresponding to $m_{a'}$. The loop on line \ref{getneighbours} will then consider $m_{a'}$ and $ca'$ and, since $a'$ matches the mask as shown above and is available, $a'$ will be added to $S_i$, and thus upon the exit of $BACSS_i$, $a' \in S_i$.
\item  A mask subsuming $mask_{a}^{ca}$ has already been used with $f$. We show hereafter that this case cannot occur under the assumptions above, i.e. that $a'$ is still available when $a$ is considered for expansion. Indeed, let $msk'$ be the mask subsuming $mask_{a}^{ca}$, that has already been registered as used in conjunction with $f$. Then there must have been an assignment $a''$ of rule $m_{a''}$ such that $msk'$ comes from a conflict area $ca'' = <(v''_1,\dots,v''_n), f>$ of $m_{a''}$. This means that (i) $<ca',ca''>$ is a conflict between $m_{a'}$ and $m_{a''}$ and that furthermore (ii) $a'$ matches $msk'$ on $ca'$. Indeed, the last point is a direct consequence of Proposition \ref{prop:masks} and the fact that $a'$ matches $mask_{a}^{ca}$ on $ca'$ and $msk'\succeq mask_{a}^{ca}$. Reasoning similarly as above, it then follows that the search due to $a''$ must have added $a'$ to $S_i$, or else $a'$ was already in $S_i$ prior to the expansion of $a''$. Either way, this means that $a'$ is already in $S_i$ when $a$ is considered for expansion, which contradicts our assumption of $a'$ being still available at the same execution stage. 
\end{itemize}

\end{proof}

\noindent \textbf{Chase failure cases}. Note that in the above reasoning on correctness we have not yet considered chase failure cases. These can in fact be handled straightforwardly by noting that the following holds:
\begin{proposition}\label{prop:failsubset}
Let $A$ be an assignment set and $S \subseteq A$. Let $ChSeq=s_1, \dots, s_n$ be a failed chase sequence on $S$. Then $ChSeq$ is a failed chase sequence on $A$.
\end{proposition}

\noindent \textbf{Putting it all together}. The results presented above now allow us to synthetically complete the proof of Theorem \ref{TH:COR}. Indeed, Lemma \ref{lemma:point2} allows us to infer that any successful run of our algorithm creates the conditions of application of Lemma \ref{lemma:point1}. Furthermore:
\begin{itemize}
\item If the Oblivious Chase fails (i.e. any Oblivious Chase sequence fails), then our algorithms will exit with failure on any run. Indeed, if there existed a successful run, by Lemmas\ref{lemma:point2} and \ref{lemma:point1} we would be able to provide a successful terminating chase sequence.
\item If the Oblivious Chase does not fail (i.e. no Oblivious Chase sequence fails), then our algorithms cannot exit with failure. Indeed, if a failed run was possible, then by Proposition \ref{prop:failsubset} we would be able to provide a failed chase sequence (by using the failed chase sequence of the $BACSS$ call that raised the {\em ChaseFail} exception). 

Since our algorithms succeed, let $J_{IC/P}$ be the output of a run of either version of our algorithms, that is, the \ouralgo\ or the \ouralgo\ With Parallelization. By Lemmas \ref{lemma:point2} and \ref{lemma:point1}, there exists a successful Oblivious Chase sequence s.t., denoting by $J_{seq}$ its result, $J_{seq} = J_{IC/P}$. Using the same shortcut of terminology as in our Theorem statement and relying on Proposition \ref{prop:iso}, we then conclude that $J_{IC/P}$ is isomorphic to the result of the Oblivious Chase, thus completing the proof of Theorem           \ref{TH:COR}
\end{itemize}
\end{document}